\documentclass{article}
\usepackage[usenames,dvipsnames,dvinames]{xcolor}
\usepackage{multirow}
\usepackage{microtype}
\usepackage{enumitem}
\usepackage{amsmath,amsthm}
\usepackage[mathscr]{eucal}
\usepackage{amssymb}
\usepackage{footnote}
\usepackage{graphicx} 
\usepackage[vlined,ruled]{algorithm2e}
\usepackage{algorithmic}
\usepackage{wrapfig}
\usepackage{authblk}

\usepackage{float}

\newcommand{\union}{\cup}

\newcommand{\mmin}{MMin}
\newcommand{\mmax}{MMax}
\newcommand{\mmini}{MMin} 
\newcommand{\ggrow}{\hat{g}} 
\newcommand{\gshrink}{\check{g}}  
\newcommand{\gthree}{\bar{g}} 

\newcommand{\lovasz}{Lov\'asz}

\DeclareMathOperator*{\argmax}{argmax}
\DeclareMathOperator*{\argmin}{argmin}

\newtheorem{theorem}{Theorem}[section]

\newtheorem{lemma}[theorem]{Lemma}

\newtheorem{corollary}[theorem]{Corollary}

\providecommand{\doarxiv}{true}
\usepackage{ifthen}
\newboolean{isarxiv}
\setboolean{isarxiv}{\doarxiv} 
\ifthenelse{\boolean{isarxiv}}{%
\newcommand{\arxiv}[1]{#1}
\newcommand{\notarxiv}[1]{}
}{
\newcommand{\arxiv}[1]{}
\newcommand{\notarxiv}[1]{#1}
}
\newcommand{\arxivalt}[2]{\ifthenelse{\boolean{isarxiv}}{#1}{#2}}

\usepackage{ifthen}
\newboolean{isdraft}
\setboolean{isdraft}{false} 
\ifthenelse{\boolean{isdraft}}{%
\newcommand{\myaddcomment}[3]{{\color{#1}{\ensuremath{\langle\!\!\langle}{\bf {#2} :} {#3}\ensuremath{\rangle\!\!\rangle}}}}
\newcommand{\rishabh}[1]{\myaddcomment{orange}{Rishabh}{#1}}
\newcommand{\JTR}[1]{\myaddcomment{orange}{Jeff\ensuremath{\rightarrow}Rishabh}{#1}}
\newcommand{\STR}[1]{\myaddcomment{orange}{Steffi\ensuremath{\rightarrow}Rishabh}{#1}}
\newcommand{\jeff}[1]{\myaddcomment{blue}{Jeff}{#1}}
\newcommand{\RTJ}[1]{\myaddcomment{blue}{Rishabh\ensuremath{\rightarrow}Jeff}{#1}}
\newcommand{\STJ}[1]{\myaddcomment{blue}{Steffi\ensuremath{\rightarrow}Jeff}{#1}}
\newcommand{\steffi}[1]{\myaddcomment{red}{Steffi}{#1}}

}{
\newcommand{\rishabh}[1]{}
\newcommand{\JTR}[1]{}
\newcommand{\jeff}[1]{}
\newcommand{\RTJ}[1]{}
\newcommand{\toboth}[1]{}
\newcommand{\steffi}[1]{}
\newcommand{\STR}[1]{}
\newcommand{\STJ}[1]{}
}

\newcommand{\Cs}{\mathcal{C}}

\newcounter{propcounter}
\usepackage[usenames,dvipsnames,dvinames]{xcolor}
\usepackage{hyperref}

\title{Fast Semidifferential-based Submodular Function Optimization \thanks{A shorter version of this appeared in Proc. International Conference of Machine Learning (ICML), Atlanta, 2013}}
\author[1]{Rishabh Iyer\thanks{rkiyer@u.washington.edu}}
\author[2]{Stefanie Jegelka\thanks{stefje@eecs.berkeley.edu}}
\author[1]{Jeff Bilmes\thanks{bilmes@u.washington.edu}}

\affil[1]{Department of EE, University of Washington, Seattle}
\affil[2]{Department of EECS, University of California, Berkeley}

\begin{document}
\maketitle
\begin{abstract}
  We present a practical and powerful new framework for both
  unconstrained and constrained submodular function optimization based
  on discrete semidifferentials (sub- and super-differentials).
  The resulting algorithms, which repeatedly compute and then efficiently
  optimize submodular semigradients, offer new and generalize
  many old methods for submodular optimization.  Our approach,
  moreover, takes steps towards providing a unifying paradigm
  applicable to both submodular minimization and maximization,
  problems that historically have been treated quite distinctly. The
  practicality of our algorithms is important since
  interest in submodularity, owing to its natural and wide applicability,
  has recently been in ascendance within machine
  learning.  We analyze theoretical properties of our algorithms for
  minimization and maximization, and show that many state-of-the-art
  maximization algorithms are special cases. Lastly, we complement our
  theoretical analyses with supporting empirical
  experiments. 
\end{abstract}

\section{Introduction}

In this paper, we address minimization and maximization problems
of the following form:
\begin{align}\label{probsmain}
\mbox{Problem 1: } \min_{X \in \mathcal C} f(X),\qquad
\mbox{Problem 2: } \max_{X \in \mathcal C} f(X) \nonumber
\end{align}
where $f: 2^V \to \mathbb R$ is a discrete set function on subsets of a ground
set $V = \{1, 2, \cdots, n\}$, and $\Cs \subseteq 2^V$ is a family of
feasible solution sets. The set $\Cs$ could
express, for example, that solutions must be an independent set in a
matroid, a limited budget knapsack, or a cut (or spanning tree, path, or
matching) in a graph.  Without making any further assumptions about
$f$, the above problems are trivially worst-case exponential time and
moreover inapproximable.

If we assume that $f$ is submodular, however, then in many cases the
above problems can be approximated and in some cases solved
exactly in polynomial time.  A function $f: 2^V \to \mathbb R$ is said
to be \emph{submodular} \cite{fujishige2005submodular} 
if for all
subsets $S, T \subseteq V$, it holds that $f(S) + f(T) \geq f(S \cup
T) + f(S \cap T)$.
Defining $f(j | S) \triangleq f(S \cup j) - f(S)$ as the gain of $j\in
V$ with respect to $S \subseteq V$, then $f$ is submodular if and only if 
$f(j | S) \geq f(j | T)$ for all $S \subseteq T$ and $j \notin T$.
Traditionally, submodularity has been a key structural property for
problems in combinatorial optimization, and for applications in
econometrics, circuit and game theory, and operations research.  More
recently, 
submodularity's popularity
in machine learning has been on the rise.

On the other hand, a potential stumbling block is that machine
learning problems are often large (e.g., ``big data'') and are getting
larger.  For general unconstrained submodular minimization, the
computational complexity often scales as a high-order
polynomial. These algorithms are designed to solve the most general
case and the worst-case instances are often contrived and unrealistic.
Typical-case instances are much more benign, so simpler algorithms
(e.g., graph-cut) might suffice. In the constrained case, however, the
problems often become NP-complete.  Algorithms for submodular
maximization are very different in nature from their submodular
minimization cohorts, and their complexity too varies depending on the
problem.  In any case, there is an urgent need for efficient,
practical, and scalable algorithms for the aforementioned problems if
submodularity is to have a lasting impact on the field of machine
learning.

In this paper, we address the issue of scalability and simultaneously draw connections across the apparent gap between minimization and maximization problems.
We demonstrate that many 
algorithms for
submodular maximization may be viewed as special cases of a generic
minorize-maximize framework that relies on discrete semidifferentials.
This framework encompasses state-of-the-art greedy and local search techniques, and provides a rich class of very practical algorithms. In addition, we show that any approximate submodular maximization algorithm can be seen as an instance of our framework.

We also present a complementary majorize-minimize framework for
submodular minimization that makes two contributions.
For unconstrained minimization,
we obtain new
nontrivial bounds on the lattice of minimizers, thereby reducing the
possible space of candidate minimizers. This method easily integrates
into any other exact minimization algorithm as a preprocessing step to
reduce running time. In the constrained case, we obtain practical
algorithms with bounded approximation factors.
We observe these
algorithms to be empirically competitive to more complicated ones.

As a whole, the semidifferential framework offers a new unifying perspective and basis for
treating submodular minimization and maximization
problems in both the constrained and unconstrained case.  While it has
long been known \cite{fujishige2005submodular} that submodular
functions have tight subdifferentials, our results rely
on a recently discovered property \cite{rkiyersubmodBregman2012,
  jegelka2011-nonsubmod-vision} showing that submodular functions also
have superdifferentials. Furthermore, our approach is entirely
combinatorial, thus complementing (and sometimes obviating) related
relaxation methods. \looseness-1

\section{Motivation and Background}
Submodularity's escalating popularity in machine learning is due to
its natural applicability. Indeed, instances of Problems 1 and 2 are
seen in many forms, to wit:

\paragraph{MAP inference/Image segmentation:} 
Markov Random Fields with pairwise attractive potentials are important
in computer vision, where MAP inference is identical to unconstrained
submodular minimization solved via minimum cut
\cite{boykovJolly01}. 
A richer higher-order
model can be induced for which
MAP inference corresponds to
Problem 1 where $V$ is a set of edges in a graph, 
and $\mathcal C$ is a set of cuts in this graph
--- this was shown to significantly improve
many image segmentation results \cite{jegelka2011-nonsubmod-vision}.
Moreover, \cite{delong2012minimizing} efficiently solve MAP inference
in a sparse higher-order graphical model by restating the problem as a
submodular vertex cover, i.e., Problem 1 where $\Cs$ is the set of all
vertex covers in a graph.\looseness-1

\paragraph{Clustering:} Variants of submodular minimization have been
successfully applied to clustering problems
\cite{narasimhan2006q,nagano10macc}. 

\paragraph{Limited Vocabulary Speech Corpora:} The problem of finding
a maximum size speech corpus with bounded vocabulary~\cite{lin11}
can be posed as submodular function minimization 
subject to a size constraint. Alternatively, cardinality
can be treated as a penalty, reducing the problem to unconstrained
submodular minimization~\cite{jegelkanips}.

\arxiv{
\paragraph{Size constraints:} The densest $k$-subgraph and
size-constrained graph cut problems correspond to submodular
minimization with cardinality constraints,
problems that are very hard~\cite{svitkina2008submodular}.
Specialized algorithms for cardinality and related constraints were proposed e.g.\ in \cite{svitkina2008submodular,nagano2011}. 
}


\paragraph{Minimum Power Assignment:}
In wireless networks, one seeks a connectivity structure that maintains connectivity at a minimum energy consumption.
This problem is equivalent to finding a
suitable structure (e.g., a spanning tree) minimizing a submodular
cost function \cite{wan02networks}.

\paragraph{Transportation:} Costs in real-world transportation
problems are often non-additive. For example, it may be cheaper to
take a longer route owned by one carrier
rather than a shorter route that switches carriers.  Such economies of
scale, or ``right of usage'' properties are captured in the
``Categorized Bottleneck Path Problem'' -- a shortest path problem
with submodular costs \cite{averbakh94}. Similar costs have been
considered for spanning tree and matching problems.


\paragraph{Summarization/Sensor placement:} 
Submodular maximization also arises in many subset
extraction problems. 
Sensor placement \cite{krause2008near}, document
summarization \cite{linacl} and speech data subset
selection~\cite{lin2009select}, for example,
 are instances of  submodular maximization.

\paragraph{Determinantal Point Processes: } The Determinantal Point Processes (DPP’s) which have found numerous applications in machine learning~\cite{kulesza2012determinantal} are known to be log-submodular distributions. In particular, the MAP inference problem is a form of non-monotone submodular maximization.
\vspace{2ex}

Indeed, there is strong motivation for solving Problems 1 and 2 but,
as mentioned above, these problems come not without computational
difficulties.  Much work has therefore been devoted to developing
optimal or near optimal algorithms. Among the several algorithms
\cite{mccormick2005submodular}
for the unconstrained variant of Problem 1, where $\Cs = 2^V$, the
best complexity to date is $O(n^5 \gamma +
n^6)$~\cite{orlin2009faster} ($\gamma$ is the cost of evaluating
$f$). This has motivated studies on faster, possibly special case or
approximate, methods \cite{stobbe10efficient,jegelkanips}. Constrained
minimization problems, even for simple constraints such as a cardinality lower
bound, are mostly NP-hard, and not approximable to within better than a
polynomial factor. Approximation algorithms for these problems with
various techniques have been studied in
\cite{svitkina2008submodular,iwata2009submodular,goel2009approximability,jegelka2011-inference-gen-graph-cuts}.
Unlike submodular minimization, all forms of submodular maximization
are NP-hard. Most such problems, however, admit constant-factor
approximations,
which are attained via very simple combinatorial algorithms
\cite{nemhauser1978,feldman2012optimal}. \looseness-1

Majorization-minimization (MM)\footnote{MM also refers to
  minorization-maximization here.} algorithms are known to be useful in
machine learning~\cite{hunter2004tutorial}. Notable examples include
the EM algorithm~\cite{mclachlan1997algorithm} and the convex-concave
procedure~\cite{yuille2002concave}.  Discrete instances have been used
to minimize the difference between submodular
functions~\cite{narasimhanbilmes,rkiyeruai2012}, but these algorithms
generally lack theoretical guarantees.  This paper shows, by contrast,
that for submodular optimization, MM algorithms have strong
theoretical properties and empirically work very well.  \looseness-1


\section{Submodular semi-differentials}

We first briefly introduce submodular
semidifferentials. Throughout this paper, we assume normalized
submodular functions (i.e., $f(\emptyset) = 0$).
The subdifferential $\partial_f(Y)$ of a submodular set function $f:
2^V \to \mathbb{R}$ for a set $Y \subseteq V$ is defined 
\cite{fujishige2005submodular}
analogously
to the subdifferential of a continuous convex function:
\begin{align}
\partial_f(Y) &= \{y \in \mathbb{R}^n:\\
\nonumber
&\quad f(X) - y(X) \geq f(Y) - y(Y)\;\text{for all } X \subseteq V\} 
\end{align}
For a vector $x \in \mathbb{R}^V$ and $X \subseteq V$, we write
$x(X) = \sum_{j \in X} x(j)$ --- in such case, we say that $x$ is a normalized
\emph{modular} function. 
We shall denote a subgradient at $Y$ by $h_Y \in \partial_f(Y)$. The extreme points of $\partial_f(Y)$ may be computed via a greedy algorithm: Let $\sigma$ be a permutation of $V$ that assigns the elements in $Y$ to the first $|Y|$ positions ($\sigma(i) \in Y$ if and only if $i \leq |Y|$). \arxiv{
  \begin{figure}
    \centering
    \includegraphics[width=0.35\textwidth]{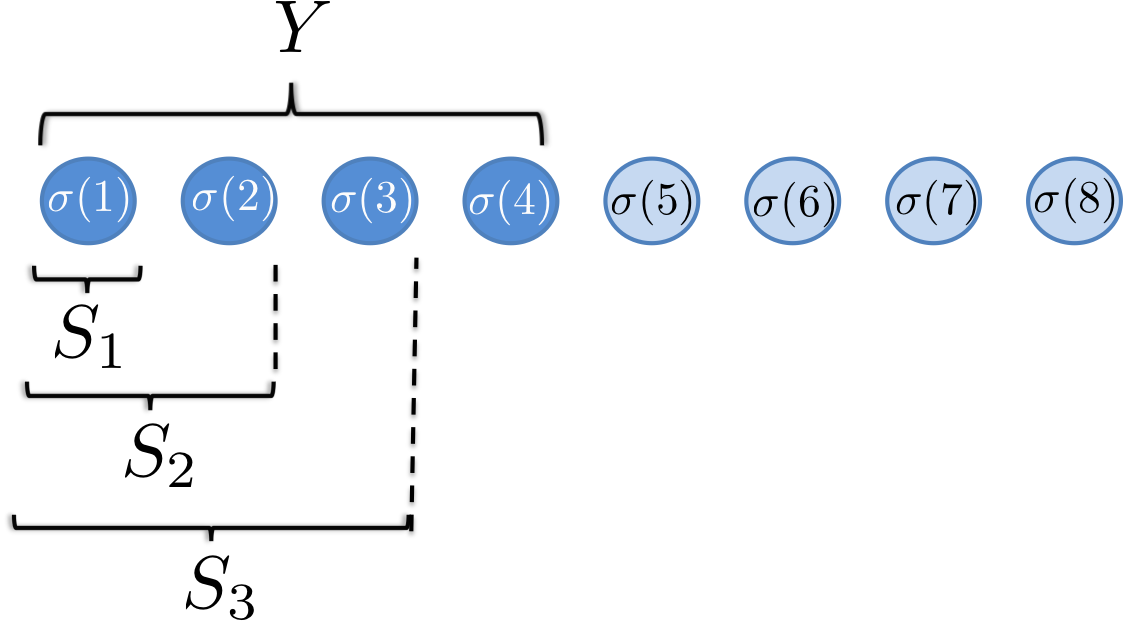}
\caption{Illustration of the chain of sets and permutation $\sigma$}
\end{figure}}
\JTR{check permutation notation consistency through entire paper.}\RTJ{fixed it above}
Each such permutation defines a chain with elements $S_0^\sigma = \emptyset$, 
$S^{\sigma}_i = \{ \sigma(1), \sigma(2), \dots, \sigma(i) \}$ and
$S^{\sigma}_{|Y|} = Y$. This chain defines an extreme point
$h^{\sigma}_Y$ of $\partial_f(Y)$ with entries
\begin{align}
h^{\sigma}_Y(\sigma(i)) = 
f(S^{\sigma}_i) - f(S^{\sigma}_{i-1}). 
\end{align}
%

Surprisingly, we can also define superdifferentials 
$\partial^f(Y)$ of a submodular
function \cite{jegelka2011-nonsubmod-vision,rkiyersubmodBregman2012} at $Y$: 
\begin{align}\label{supdiff-def}
\partial^f(Y) &= \{y \in \mathbb{R}^n:\\
&\quad f(X) - y(X) \leq f(Y) - y(Y);\text{for all } X \subseteq V\}  \nonumber
\end{align}
We denote a generic supergradient at $Y$ by $g_Y$. It is easy to show
that the polyhedron $\partial^f$ is non-empty.
We define three special
supergradients $\ggrow_Y$ (``grow''), $\gshrink_Y$ (``shrink'') and $\gthree_Y$
as follows~\cite{rkiyersubmodBregman2012}: 
  \begin{align*}
    &\ggrow_Y(j) = f(j \mid V \setminus \{j\}) & \ggrow_Y(j) = f(j\mid Y)\\
    &\gshrink_Y(j) = f(j \mid Y \setminus \{j\}) & \gshrink_Y(j) = f(j\mid \emptyset)\\
    &\underbrace{\gthree_Y(j) = f(j \mid V \setminus \{j\})} & \underbrace{\gthree_Y(j) = f(j\mid\emptyset)}\\
    &\qquad\text{ for } j \in Y & \text{for } j \notin Y.
  \end{align*}
\arxiv{For a \emph{monotone} submodular function, i.e., a function satisfying $f(A) \leq f(B)$ for all $A \subseteq B \subseteq V$, the sub- and supergradients defined here are nonnegative.}

\section{The discrete MM framework} 

With the above semigradients, we can define a generic MM
algorithm. In each iteration, the algorithm optimizes a modular
approximation formed via the current solution $Y$. For minimization, we
use an upper bound
\begin{equation}
  \label{eq:1}
  m^{g_Y}(X) = f(Y) + g_Y(X) - g_Y(Y)\; \geq f(X),
\end{equation}
and for maximization a lower bound 
\begin{equation}
  \label{eq:2}
  m_{h_Y}(X) = f(Y) + h_Y(X) - h_Y(Y)\; \leq f(X).
\end{equation}
Both these bounds are tight at the current solution, satisfying $m_{g_Y}(Y) = m_{h_Y}(Y) = f(Y)$.
In almost all cases, optimizing the modular approximation is much faster than optimizing the original cost function $f$.

\begin{algorithm}
\caption{Semigradient Descent Algorithm}
\begin{algorithmic}[1]
\STATE Start with an arbitrary $X^0$.
\REPEAT
\STATE Pick a semigradient $h_{X^t}$ {[ $g_{X^t}$]} at $X^t$ 
\STATE $X^{t+1}:= \argmax_{X \in \mathcal C} m_{h_{X^t}}(X)$ \\
{[ $X^{t+1}:= \argmin_{X \in \mathcal C} m^{g_{X^t}}(X)$]}
\STATE $t \leftarrow t+1$
\UNTIL{we have converged $(X^{i-1} = X^i)$}
\end{algorithmic}
\label{mirroropt}
\end{algorithm}

Algorithm~\ref{mirroropt} shows our discrete MM scheme for
maximization (\mmax{}) [and minimization (\mmin{})] , and for
both constrained and unconstrained settings.  Since we are
minimizing a tight upper bound, or maximizing a tight lower bound, the algorithm must make progress.
\begin{lemma}
Algorithm~\ref{mirroropt} monotonically improves the objective function value for Problems 1 and 2 at every iteration, as long as a linear function can be exactly optimized over $\mathcal C$.
\end{lemma}\looseness-1
\arxiv{\begin{proof}
By definition, it holds that $f(X^{t+1}) \leq m^{g_{X^t}}(X^{t+1})$. Since $X^{t+1}$ minimizes $m^{g_{X^t}}$, it follows that
\begin{equation}
  \label{eq:4}
  f(X^{t+1}) \leq m^{g_{X^t}}(X^{t+1}) \leq m^{g_{X^t}}(X^t) = f(X^t).
\end{equation}
The observation that Algorithm~\ref{mirroropt} monotonically increases the objective of maximization problems follows analogously.
\end{proof}}

Contrary to standard continuous subgradient descent schemes,
Algorithm~\ref{mirroropt} produces a feasible
solution at each iteration, thereby circumventing any rounding or projection steps that
might be challenging under certain types of constraints.  In addition,
it is known that for relaxed instances of our problems, subgradient
descent methods can suffer from slow convergence~\cite{frbach1}.
Nevertheless, 
Algorithm~\ref{mirroropt} 
still relies on the choice of the semigradients
defining the bounds.
Therefore, we next analyze the effect of certain choices of
semigradients. 


\section{Submodular function minimization}\label{sec:minimization}

For minimization problems, we use \mmin{} with the supergradients
$\ggrow_X, \gshrink_X$ and $\gthree_X$. In both the unconstrained and
constrained settings, this yields a number of new approaches to
submodular minimization.


\subsection{Unconstrained Submodular Minimization}\label{sec:unconstrained_min}
We begin with unconstrained minimization, where $\mathcal{C}=2^V$ in Problem 1.
Each of the three supergradients yields a different variant of Algorithm~\ref{mirroropt}, and we will call the resulting algorithms \mmini-I, II and III, respectively.
We make one more assumption: of the minimizing arguments in Step 4 of Algorithm~\ref{mirroropt}, we always choose a set of minimum cardinality.

\mmini-I is very similar to the algorithms proposed in~\cite{jegelkanips}. Those authors, however, decompose $f$ and explicitly represent graph-representable parts of the function $f$.
We do not require or consider such a restriction here. 

Let us define the sets $A = \{j: f(j | \emptyset) < 0\}$ and 
$B = \{j: f(j | V\setminus\{j\}) \leq  0\}$. 
Submodularity implies that $A \subseteq B$,
and this allows us to define a lattice\arxiv{\footnote{This lattice contains all sets $S$ satisfying $A \subseteq S \subseteq B$}} $\mathcal L = [A,B]$ whose
least element is the set $A$ and whose greatest element is the set is
$B$. 
This sublattice $\mathcal{L}$ of $[\emptyset,V]$ retains all
minimizers $X^*$ (i.e., $A \subseteq X^* \subseteq B$ for all $X^*$):
\begin{lemma}\cite{fujishige2005submodular} \label{fuj}
Let $\mathcal L^{*}$ be the lattice of the global minimizers of a submodular function $f$. Then $\mathcal L^* \subseteq \mathcal L$, where we use $\subseteq$ to denote a sublattice. \looseness-1
\end{lemma}
Lemma~\ref{fuj} has been used to prune down the search space of the
minimum norm point algorithm\arxiv{ from the power set of $V$ to a smaller lattice}~\cite{frbach1,fujishige2011submodular}. Indeed, $A$ and $B$ may be obtained by using \mmini-III:
 \begin{lemma}\label{lem:imaiii}
With $X^0 = \emptyset$ and $X^0 = V$, \mmini-III returns the sets $A$ and $B$, respectively.
Initialized by an arbitrary $X^0$, \mmini-III converges to $(X^0 \cap B) \cup A$. 
 \end{lemma}\looseness-1
\arxiv{\begin{proof}
When using $X^0 = \emptyset$, we obtain $X^1 = \argmin_X f(\emptyset) + \sum_{j \in X} f(j) = A$. Since $A \subseteq B$, the algorithm will converge to $X^1 = A$. 
At this point, no more elements will be added, since for all $i \notin A$ we have $\gthree_{X^1}(i) = f(i\mid \emptyset)>0$. Moreover, the algorithm will not remove any elements: for all $i \in A$, it holds that $\gthree_{X^1}(i) = f(i\mid V\setminus i) \leq f(i) \leq 0$.
By a similar argumentation, the initialization $X^0 = V$ will lead to $X^1= B$, where the algorithm terminates.
If we start with any arbitrary $X^0$, \mmini-III will remove the elements $j$ with $f(j | V \setminus j) > 0$ and add the element $j$ with $f(j | \emptyset) < 0$. Hence it will add the elements in $A$ that are not in $X^0$ and remove those element from $X^0$ that are not in $B$. Let the resulting set be $X^1$. As before, for all $i \in A$, it holds that $\gthree_{X^1}(i) = f(i\mid V\setminus i) \leq f(i) \leq 0$, so these elements will not be removed in any possible subsequent iteration. The elements $i \in X^1 \setminus A$ were not removed, so $f(i\mid V\setminus i) \leq 0$. Hence, no more elements will be removed after the first iteration. Similarly, no elements will be added since for all $i \notin X^1$, it holds that $f(i\mid \emptyset) \geq f(i\mid V\setminus i) > 0$.
\end{proof}}
Lemma~\ref{lem:imaiii}
 implies that \mmini-III effectively provides a contraction of the initial lattice to $\mathcal L$, and, if $X^0$ is not in $\mathcal{L}$
, it returns a set in $\mathcal L$. Henceforth, we therefore assume that we 
start with a set $X^0 \in \mathcal L$. 

While the known lattice $\mathcal{L}$ has proven useful for warm-starts, \mmini-I and II enable us to prune $\mathcal{L}$ even further. 
%
Let $A_+$ be the set obtained by starting \mmini-I at $X^0 = \emptyset$, and $B_+$ be the set obtained by starting \mmini-II at $X^0 = V$. This yields a new, smaller sublattice $\mathcal L_+ = [A_+, B_+]$ that retains all minimizers:
\begin{theorem}\label{thm:prune}
For any minimizer $X^* \in \mathcal L$, it holds that $A \subseteq A_+ \subseteq X^* \subseteq B_+ \subseteq B$. Hence $\mathcal L^* \subseteq \mathcal L_+ \subseteq \mathcal L$. Furthermore, when initialized with $X^0 = \emptyset$ and $X^0 = V$, respectively, both \mmini-I and II converge in $O(n)$ iterations to a local minimum of $f$.
\end{theorem}\label{minthm}
By a local minimum, we mean a set $X$ that satisfies $f(X) \leq f(Y)$ for any set $Y$ that differs from $X$ by a single element.
We point out that Theorem~\ref{thm:prune} generalizes part of Lemma 3 in~\cite{jegelkanips}. \looseness-1
\arxiv{
  \begin{figure}
    \centering
    \includegraphics[width=0.35\textwidth]{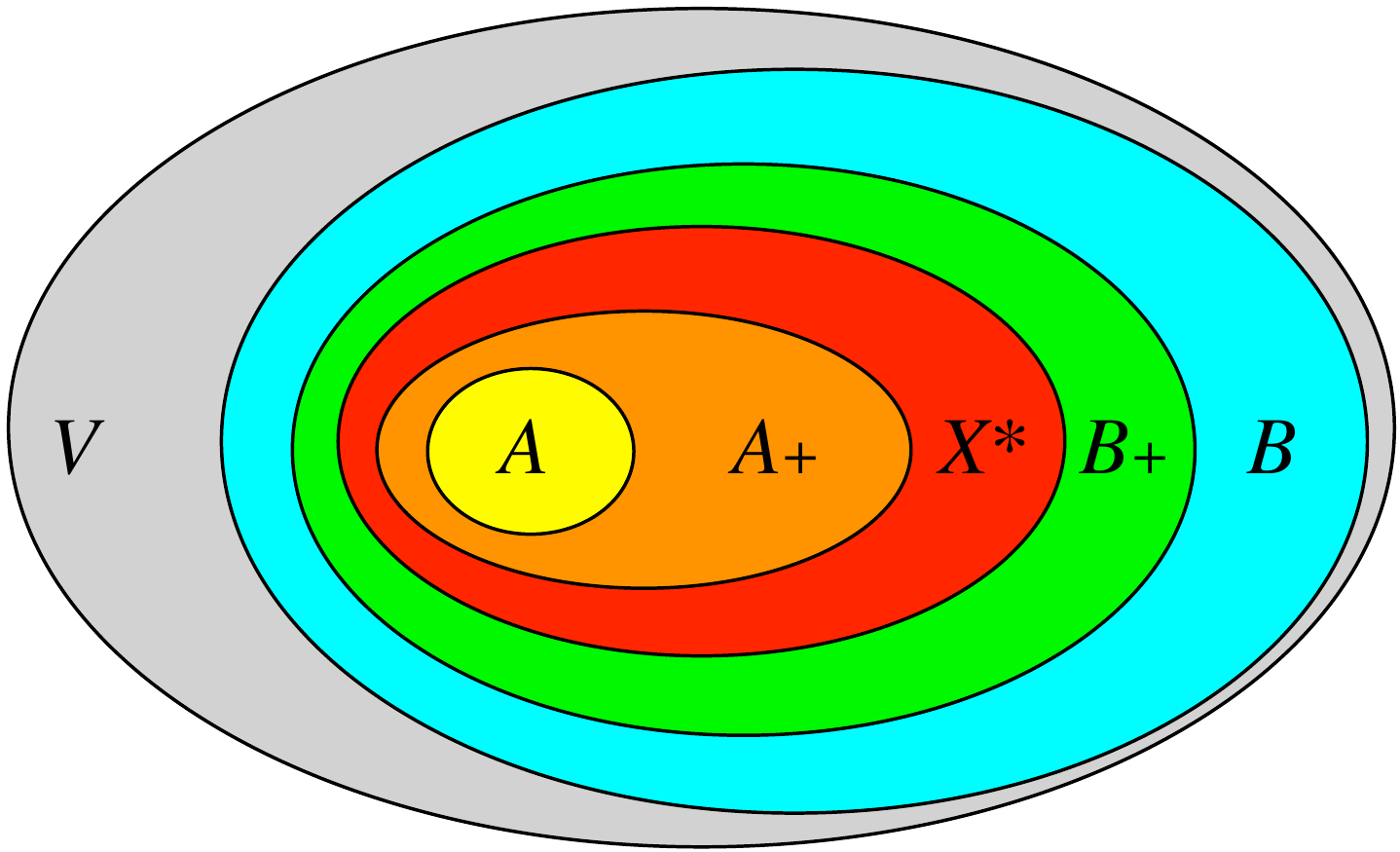}
    \caption{Venn diagram for the lattices obtained by \mmini-I, II and III. We are searching for the optimal set $X^* \subseteq V$. The lattice $\mathcal{L}$ contains all sets $S$ ``between'' $A$ and $B$, i.e., $A \subseteq S \subseteq B$. The lattice $\mathcal{L}_+$ uses the sets $A_+$ and $B_+$ instead (it contains all sets $T$ with $A_+ \subseteq T \subseteq B_+$) and therefore provides a tighter bound and smaller search space around the optimal solution $X^*$.}
    \label{fig:illustrate_lattices}
  \end{figure}
For the proof, we build on the following Lemma:
\begin{lemma}\label{lem:mmin_iter}
Every iteration of \mmini-I can be written as $X^{t+1} = X^t \cup \{j: f(j | X^t) < 0\}$. Similarly, 
every iteration of \mmini-II can be expressed as $X^{t+1} = X^t \backslash \{j: f(j | X^t \setminus j) > 0\}$.   
\STR{Independent of initialization, right?}
\end{lemma}
\begin{proof} \emph{(Lemma~\ref{lem:mmin_iter})}
Throughout this paper, we assume that we select only the minimal minimizer of the modular function at every step. In other words, we do not choose the elements that have zero marginal cost. 
We observe 
that in iteration $t+1$ of \mmini-I, we add the elements $i$ with $\ggrow_{X^t}(i) < 0$, i.e., $X^{t+1} = X^t \cup \{j: f(j | X^t) < 0\}$. No element will ever be removed, since $\ggrow_{X^t}(i) = f(i \mid V\setminus i) \leq f(i\mid X^{t-1}) \leq 0$.
If we start with $X^0 = \emptyset$, then after the first iteration, it holds that $X^1 = \argmin_X f(\emptyset) + \sum_{j \in X} f(j)$. Hence $X^1 = A$.
\mmini-I terminates when reaching a set $A_+$, where $f(j | A_+) \geq 0$, for all $j \notin A_+$.

The analysis of \mmini-II is analogous. 
In iteration $t+1$, we remove the elements $i$ with $\gshrink_{X^t}(i) > 0$, i.e., $X^{t+1} = X^t \backslash \{j: f(j | X^t - j) > 0\}$. Similarly to the argumentation above, \mmini-II never adds any elements. If we begin with $X^0 = V$, then $X^1 = \mbox{arg min}_X f(V) + \sum_{j \in V \backslash X} f(j | V - \{j\})$, and therefore $X^1 = B$. \mmini-II terminates with a set $B_+$.
\end{proof}
Now we can prove Theorem~\ref{thm:prune}.
\begin{proof}{\emph{(Thm.~\ref{thm:prune})}}
Since, by Lemma~\ref{lem:mmin_iter}, \mmini-I only adds elements and \mmini-II only removes elements, at least one in each iteration, both algorithms terminate after $O(n)$ iterations.

Let us now turn to the relation of $X^*$ to $A$ and $B$. Since $f(i)<0$ for all $i \in A$, the set $X^1=A$ found in the first iteration of \mmin-I must be a subset of $X^*$. Consider any subset $X^t \subseteq X^*$. Any element $j$ for which $f(j\mid X^t) < 0$ must be in $X^*$ as well, because by submodularity, $f(j\mid X^*) \leq f(j\mid X^{t}) < 0$. This means $f(X^* \cup j) < f(X^*)$, which would otherwise contradict the optimality of $X^*$. The set of such $j$ is exactly $X^{t+1}$, and therefore $X^{t+1} \subseteq X^*$.
This induction shows that \mmini-I, whose first solution is $A \subseteq X^*$, always returns a subset of $X^*$.
Analogously, $B \supseteq X^*$, and \mmini-II only removes elements $j \notin X^*$.


Finally, we argue that $A_+$ is a local minimum; the proof for $B_+$ is analogous.
Algorithm \mmini-I generates a chain $\emptyset = X^0 \subseteq X^1 \subseteq X^2 \cdots \subseteq A_+ = X^T$. For any $t \leq T$, consider $j \in X^t\setminus X^{t-1}$. Submodularity implies that $f(j | A_+ \setminus j) \leq f(j | X^{t-1}) < 0$. The last inequality follows from the fact that $j$ was added in iteration $t$. Therefore, removing any $j \in A_+$ will increase the cost. Regarding the elements $i \notin A_+$, we observe that \mmini-I has terminated, which implies that $f(i \mid A_+) \geq 0$. Hence, adding $i$ to $A_+$ will not improve the solution, and $A_+$ is a local minimum.
%
\end{proof}}\looseness-1

Theorem~\ref{thm:prune} has a number of nice implications. First, it provides a tighter bound on the lattice of minimizers of the submodular function $f$ that, to the best of our knowledge, has not been used or mentioned before. 
\arxiv{The sets $A_+$ and $B_+$ obtained above are guaranteed to be supersets and subsets of $A$ and $B$, respectively, as illustrated in Figure~\ref{fig:illustrate_lattices}. }%
This means we can start any algorithm for submodular minimization from the lattice $\mathcal L_+$ instead of the initial lattice $2^V$ or $\mathcal{L}$. When using an algorithm whose running time is a high-order polynomial of $|V|$, any reduction of the ground set $V$ is beneficial. 
Second, each iteration of \mmini{} takes linear time. Therefore, its total running time is $O(n^2)$.
Third, Theorem~\ref{thm:prune} states that both \mmini-I and II converge to a local minimum. \arxiv{This may be counter-intuitive if one considers that each algorithm either only adds or only removes elements. }In consequence, a local minimum of a submodular function can be obtained in $O(n^2)$, a fact that is of independent interest and that does not hold for local maximizers~\cite{janvondrak}. \STR{They show an upper bound, not a lower bound on the running time right? So we could say that this \emph{probably} does not hold for maximizers?}

The following example illustrates that $\mathcal L_+$ can be a strict
subset of $\mathcal L$ and therefore provides non-trivial pruning.
Let $w_1, w_2 \in \mathbb{R}^V$, $w_1 \geq 0$ be two vectors, each
defining a linear (modular) function. Then the function $f(X) =
\sqrt{w_1(X)} + w_2(X)$ is submodular. Specifically, let $w_1 = [3, 9,
17, 14, 14, 10, 16, 4, 13, 2]$ and $w_2 = [-9, 4, 6, -1, 10, -4 , -6,
-1, 2, -8]$. Then we obtain $\mathcal{L}$ defined by $A = [1,6,7,10]$
and $B = [1,4,6,7,8,10]$. The tightened sublattice contains exactly
the minimizer: $A_+ = B_+ = X^* = [1,6,7, 8, 10]$.

\arxivalt{
As a refinement to Theorem~\ref{thm:prune}, we can show that \mmini-I and \mmini-II converge to the local minima of lowest and highest cardinality, respectively.
\begin{lemma}\label{localminall}
The set $A_+$ is the smallest local minimum of $f$ (by cardinality), and $B_+$ is the largest.
Moreover, every local minimum $Z$ is in $\mathcal{L}_+$: $Z \in \mathcal{L}_+$ for every local minimum $Z$.
\end{lemma}
\begin{proof}
The proof proceeds analogously to the proof of Theorem~\ref{thm:prune}. Let $Y_s$ be the local minimum of smallest-cardinality, and $Y_\ell$ the largest one.
First, we note that $X^0 = \emptyset \subseteq Y_s$. For induction, assume that $X^t \subseteq Y_s$. For contradiction, assume there is an element $j \in X^{t+1}$ that is not in $Y_s$. Since $j \in X^{t+1}\setminus X^t$, it holds by construction that $f(j\mid Y_s) \leq f(j \mid X^t) < 0$, implying that $f(Y_s \union j) < f(Y_s)$. This contradicts the local optimality of $Y_s$, and therefore it must hold that $X^{t+1} \subseteq Y_s$.
Consequently, $A_+ \subseteq Y_s$. But $A_+$ is itself a local minimum, and hence equality holds.
The result for $B_+$ follows analogously.


By the same argumentation as above for $Y_s$ and $Y_\ell$, we conclude
that each local minimum $Z$ satisfies $A_+ \subseteq Z \subseteq B_+$, and therefore $Z \in
\mathcal{L}_+ \subseteq \mathcal{L}$.
%
\end{proof}

As a corollary, Lemma~\ref{localminall} implies that 
if a submodular function has a unique local minimum, \mmini-I and II must find this minimum, which is a global one. 

In the following we consider two extensions of \mmini-I and II. First, we analyze an algorithm that alternates between \mmini-I and \mmini-II. While such an algorithm does not provide much benefit when started at $X^0 = \emptyset$ or $X^0 = V$, we see that with a random initialization $X^0=R$, the alternation ensures convergence to a local minimum. Second, we 
address the question of which supergradients to select in general. In particular, we show that the supergradients $\ggrow$ and $\gshrink$ subsume alternativee supergradients and provide the tightest results with \mmin{}. Hence, our results are the tight.

\paragraph{Alternating \mmini-I and II and arbitrary initializations.}
Instead of running only one of \mmini-I and II, we can run one until it stops and then switch to the other.
Assume we initialize both algorithms with a random set $X^0 = R \in \mathcal L_+$. By Theorem~\ref{thm:prune}, we know that \mmini-I will return a subset $R^1 \supset R$ (no element will be removed because all removable elements are not in $B$, and $R \subset B$ by assumption). When \mmini-I terminates, it holds that $\ggrow_{R^1}(j) = f(j | R^1) \geq 0$ for all $j \notin R^1$, and therefore $R^1$ cannot be increased using $\ggrow_{R_1}$. We will call such a set an \emph{I-minimum}. 
Similarly, \mmini-II returns a set $R_1 \subseteq R$ from which, considering that $\gshrink_{R_1}(j) = f(j | R_1 \setminus j) \leq 0$ for all $j \in R_1$, no elements can be removed. We call such a non-decreasable set a \emph{D-minimum}. Every local minimum is both an I-minimum and a D-minimum.

We can apply \mmini-II to the I-minimum $R^1$ returned by \mmini-I. Let us call the resulting set $R^2$. Analogously, applying \mmini-I to $R_1$ yields $R_2 \supseteq R_1$.


\begin{lemma}\label{arbitstart}
The sets $R_2$ and $R^2$ are local optima. Furthermore, $R_1 \subseteq R_2 \subseteq R^2 \subseteq R^1$.  
\end{lemma}
\begin{proof}
It is easy to see that $A \subseteq R_1 \subseteq B$, and $A \subseteq R^1 \subseteq B$. By Lemma~\ref{lem:mmin_iter}, \mmini-I applied to $R_1$ will only add elements, and \mmini-II on $R^1$ will only remove elements. Since $R^1$ is an I-minimum, adding an element $j \in V\setminus R^1$ to any set $X \subset R^1$ never helps, and therefore $R^1$ contains all of $R_1$, $R_2$ and $R^2$. Similarly, $R_1$ is contained in $R_2$, $R^2$ and $R^1$. In consequence, it suffices to look at the contracted lattice $[R_1,R^1]$, and any local minimum in this sublattice is a local minimum on $[\emptyset, V]$. Theorem~\ref{thm:prune} applied to the sublattice $[R_1, R^1]$ (and the submodular function restricted to the sublattice) yields the inclusion $R_2 \subseteq R^2$, so $R_1 \subseteq R_2 \subseteq R^2 \subseteq R^1$, and both $R_2$ and $R^2$ are local minima.
\end{proof}


The following lemma provides a more general view. 
\begin{lemma}\label{idmin}
Let $S_1 \subseteq S^1$ be such that $S_1$ is an I-minimum and $S^1$ is a D-minimum.
Then there exist local minima $S_2 \subseteq S^2$ in $[S_1, S^1]$ such that initializing with any $X^0 \in [S_1, S^1]$, an alternation of \mmini-I and II converges to a local minimum in $[S_2, S^2]$, and 
\begin{align}
\min_{X \in [S_1, S^1]} f(X) = \min_{X \in [S_2, S^2]} f(X).
\end{align}
\end{lemma}
\begin{proof}
Let $S_2, S^2$ be the smallest and largest local minima within $[S_1, S^1]$. By the same argumentation as for Lemma~\ref{arbitstart}, using $X^0 \in [S_1,S^1]$ leads to a local minimum within $[S_2,S^2]$. Since by definition all local optima in $[S_1, S^1]$ are within $[S_2, S^2]$, the global minimum within $[S_1, S^1]$ will also be in $[S_2, S^2]$.
\end{proof}

The above lemmas have a number of implications for minimization algorithms. First, many of the properties for initializing with $V$ or the empty set can be transferred to arbitrary initializations. In particular, the succession of \mmini-I and II will terminate in $O(n^2)$ iterations, regardless of what $X^0$ is. Second, Lemmas~\ref{arbitstart} and \ref{idmin} provide useful pruning opportunities: we can prune down the initial lattice to $[R_2,R^2]$ or $[S_2,S^2]$, respectively. In particular, if any global optimizer of $f$ is contained in $[S_1, S^1]$, it will also be contained in $[S_2, S^2]$.


\paragraph{Choice of supergradients.}
We close this section with a remark about the choice of supergradients. The following Lemma states how $\ggrow_X$ and $\gshrink_X$ subsume alternative choices of supergradients and \mmini-I and II lead to the tightest results possible.
%
\begin{lemma}\label{lem:mmin_tight}
Initialized with $X^0=\emptyset$, Algorithm~\ref{mirroropt} will converge to a subset of $A_+$ with \emph{any} choice of supergradients. Initialized with $X^0=V$, the algorithm will converge to a superset of $B_+$ with any choice of supergradients. If $X^0$ is a local minimum, then the algorithm will not move with any supergradient.
\end{lemma}
The proof of Lemma~\ref{lem:mmin_tight} is very similar to the proof of Theorem~\ref{thm:prune}.
}{
The \mmin{} algorithms can be extended to arbitrary initializations and other supergradients \cite{ouricml2013supp}. \looseness-1
}

\subsection{Constrained submodular minimization}

\mmin{} straightforwardly generalizes to constraints more complex than
$\mathcal{C} = 2^V$, and Theorem~\ref{thm:prune} still holds for more
general lattices or ring family constraints.

Beyond lattices, \mmin{} applies to any set of constraints
$\mathcal{C}$ as long as we have an efficient algorithm at hand that
minimizes a nonnegative modular cost function over $\mathcal{C}$. This
subroutine can even be approximate. Such algorithms are available for
cardinality bounds, independent sets of a matroid and many other
combinatorial constraints such as trees, paths or cuts.

As opposed to unconstrained submodular minimization, almost all cases of constrained submodular minimization are very hard \cite{svitkina2008submodular,jegelka2011-inference-gen-graph-cuts,goel2009approximability}, and admit at most approximate solutions in polynomial time.
The next theorem states an upper bound on the approximation factor achieved by \mmini-I for nonnegative, nondecreasing cost functions. An important ingredient in the bound is the \emph{curvature}~\cite{conforti1984submodular} of a monotone submodular function $f$, defined as
\begin{align}
  \kappa_f = 1 - \min\nolimits_{j \in V} f(j \mid V \backslash j)\, / \,f(j)
\end{align}

\begin{theorem}\label{SAAthm}
Let $X^* \in \argmin_{X \in \mathcal{C}}f(X)$. The solution $\widehat{X}$ returned by \mmini-I satisfies
\begin{align*}
f(\widehat{X}) \leq \frac{|X^*|}{1 + (|X^*| - 1)(1 - \kappa_f)}f(X^*) \leq \frac{1}{1 - \kappa_f}f(X^*) 
\end{align*}
If the minimization in Step 4 is done with approximation factor $\beta$, then $f(\widehat{X}) \leq \beta/(1 - \kappa_f) f(X^*)$. 
\end{theorem}\looseness-1
\arxivalt{Before proving this result, we remark that a}{A} similar, slightly looser bound was shown for cuts in~\cite{jegelka2011-nonsubmod-vision}, by using a weaker notion of curvature.
Note that the bound in Theorem~\ref{SAAthm} is at most $\frac{n}{1 + (n - 1)(1 - \kappa_f)}$, where $n=|V|$ is the dimension of the problem.
\arxivalt{
\begin{proof}
We will use the shorthand $g \triangleq \ggrow_\emptyset$.
To prove Theorem~\ref{SAAthm}, we use the following result shown in~\cite{jegelkathesis}:
\begin{equation}
  \label{eq:5}
  f(\widehat{X}) \leq \frac{g(X^*)/f(i)}{1 + (1 - \kappa_f)(g(X^*)/f(i) - 1)} f(X^*)
\end{equation}
for any $i \in V$.
We now transfer this result to curvature. To do so, we use $i^{\prime} \in \arg\max_{i \in V}f(i)$, so that $g(X^*) = \sum_{j \in X^*} f(j) \leq |X^*| f(i^{\prime})$.
Observing that the function $p(x) = \frac{x}{1 + (1 - \kappa_f)(x - 1)}$ is increasing in $x$ yields that
\begin{align}
f(\widehat{X}) \leq \frac{|X^*|}{1 + (1 - \kappa_f)(|X^*| - 1)} f(X^*).
\end{align}
%
\end{proof}

For problems where $\kappa_f < 1$, Theorem~\ref{SAAthm} yields a constant approximation factor and refines bounds for constrained minimization that are
given 
in \cite{goel2009approximability,svitkina2008submodular}. 
To our knowledge, this is the first curvature dependent bound for this general class of minimization problems. 
%
\STR{removed the detailed bounds to minimize overlap with the NIPS paper.}

A class of functions with $\kappa_f = 1$ are matroid rank functions, implying that these functions are difficult instances the \mmin{} algorithms. But several classes of functions occurring in applications have more benign curvature. For example, concave over modular functions were used in \cite{linacl, jegelka2011-nonsubmod-vision}. These comprise, for instance, functions of the form $f(X) = (w(X))^a$, for some $a \in [0, 1]$ and a nonnegative weight vector $w$, whose curvature is $\kappa_f \approx 1 - a(\frac{\min_j w(j)}{w(V)})^{1 - a} > 0$. A special case is $f(X) = |X|^a$, with curvature $\kappa_f = 1 - an^{a-1}$, or $f(X) = \log(1 + w(X))$ satisfying $\kappa_f \approx 1 - \frac{\min_j w(j)}{w(V)}$.
}{We prove Theorem~\ref{SAAthm} in~\cite{ouricml2013supp}.

  In the worst case, when $\kappa_f=1$, our approximation bounds are
  identical to prior work
  \cite{goel2009approximability,svitkina2008submodular}. 
  Matroid rank functions have $\kappa_f=1$, implying that they
  are difficult instances for \mmini. But several
  practically relevant submodular functions do satisfy $\kappa_f >
  0$. In such case, Theorem~\ref{SAAthm} replaces known polynomial
  bounds by an improved factor depending on $\kappa_f$.  An example for such functions are
  concave over modular functions used in
  \cite{stobbe10efficient,linacl,jegelka2011-nonsubmod-vision}. These
  comprise, for instance, functions of the form $f(X) = (w(X))^a$, for
  some $a \in [0, 1]$ and a nonnegative weight vector $w$, whose
  curvature is $\kappa_f \approx 1 - a(\frac{\min_j w(j)}{w(V)})^{1 - a} <
  1$. 
Another example is $f(X) = \log(1 + w(X))$ with $\kappa_f \approx
  1 - \frac{\min_j w(j)}{w(V)}$. Several applications use a sum of such functions, each with bounded support~\cite{jegelka2011-nonsubmod-vision, rkiyeruai2012}. This further reduces the curvature.}

The bounds of Theorem~\ref{SAAthm} hold after the first iteration. Nevertheless, empirically we often found that for problem instances that are not worst-case, subsequent iterations can improve the solution substantially.
Using Theorem~\ref{SAAthm}, we can bound the number of iterations the algorithm will take. To do so, we assume an $\eta$-approximate version, where we proceed only if $f(X^{t+1}) \leq (1 - \eta) f(X^t)$ for some $\eta > 0$. In practice, the algorithm usually terminates after 5 to 10 iterations for an arbitrarily small $\eta$.
%
\begin{lemma} \label{timebound}
\mmini-I runs in $O(\frac{1}{\eta} T\log \frac{n}{1 + (n-1)(1 - \kappa_f)})$ time, where $T$ is the time for minimizing a modular function subject to $X \in \mathcal C$. 
\end{lemma}
\arxiv{
\begin{proof}
At the end of the first iteration, we obtain a set $X^1$ such that $f(X^1) \leq \frac{n}{1 + (n-1)(1 - \kappa_f)}f(X^*)$. The $\eta$-approximate assumption implies that $f(X^{t+1}) \leq (1-\eta)f(X^t) \leq (1-\eta)^{t}f(X^1)$. Using that $\log(1-\eta) \leq \eta^{-1}$ and Theorem~\ref{SAAthm}, we see that the algorithm terminates after at most $O(\frac{1}{\eta}\log \frac{n}{1 + (n-1)(1 - \kappa_f)})$ iterations.
\end{proof}
}

\subsection{Experiments}
We will next see that,
apart from its theoretical properties, \mmini{} is in practice competitive to more complex algorithms. We implement and compare algorithms using Matlab and the SFO toolbox~\cite{krause2010sfo}. 
\paragraph{Unconstrained minimization} We first study the results in
Section~\ref{sec:unconstrained_min} for contracting the lattice of
possible minimizers. We measure the size of the new lattices relative to the ground set.
%
Applying \mmini-I and II (lattice $\mathcal{L}_+$) to Iwata's test function \cite{fujishige2011submodular}, we observe an average reduction of $99.5\%$ in the lattice. \mmini-III (lattice $\mathcal L$) obtains only about $60\%$ reduction. Averages are taken for $n$ between $20$ and $120$. 

In addition, we use concave over modular functions $\sqrt{w_1(X)} + \lambda w_2(V \backslash X)$ with randomly chosen vectors $w_1, w_2$ in $[0,1]^n$ and $n=50$. 
We also consider the application of selecting limited vocabulary speech corpora. 
\cite{lin11, jegelkanips} use functions of the form $\sqrt{w_1(\Gamma(X))} + w_2(V \backslash X)$, where $\Gamma(X)$ is the neighborhood function of a bipartite graph. Here, we choose $n=100$ and random vectors $w_1$ and $w_2$.
For both function classes, we vary $\lambda$ such that the optimal solution $X^*$ moves from $X^* = \emptyset$ to $X^* = V$. The results are shown in Figure~\ref{fig:unconsmin}. In both cases, we observe a significant reduction of the search space. When used as a preprocessing step for the minimum norm point algorithm (MN) \cite{fujishige2011submodular}, this pruned lattice speeds up the MN algorithm accordingly, in particular for the speech data. 
The dotted lines represent the relative time of MN including the respective preprocessing, taken with respect to MN without preprocessing. \arxiv{Figure~\ref{fig:unconsmin} also shows the average results over $10$ random choices of weights in both cases. In order to obtain accurate estimates of the timings, we run each experiment $5$ times and take the minimum of these timing valuess.}
%
\begin{figure}[t]
  \centering
\hspace{-10pt}
\includegraphics[width=0.25\textwidth]{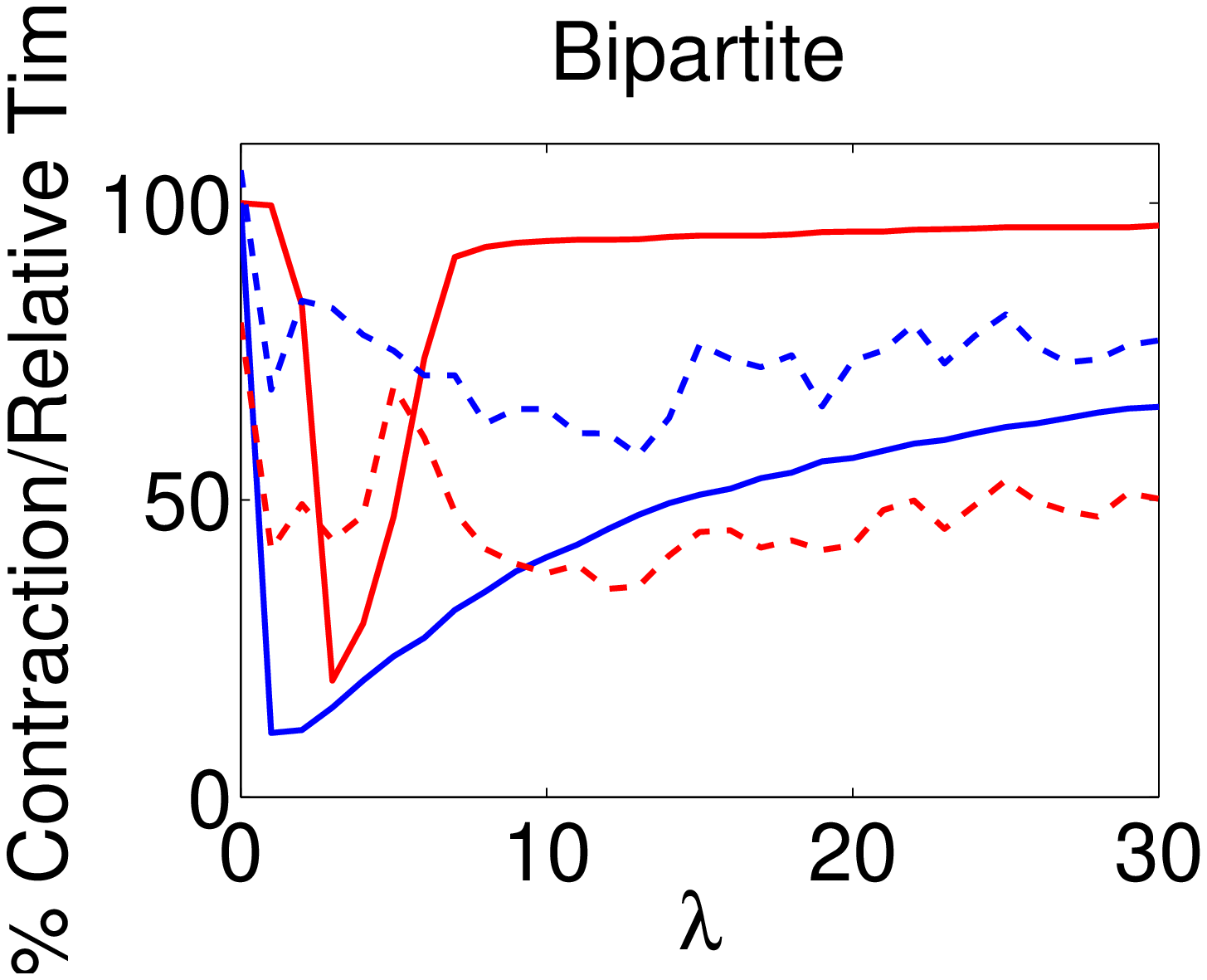}\hspace{-10pt} 
  ~ 
\includegraphics[width=0.25\textwidth]{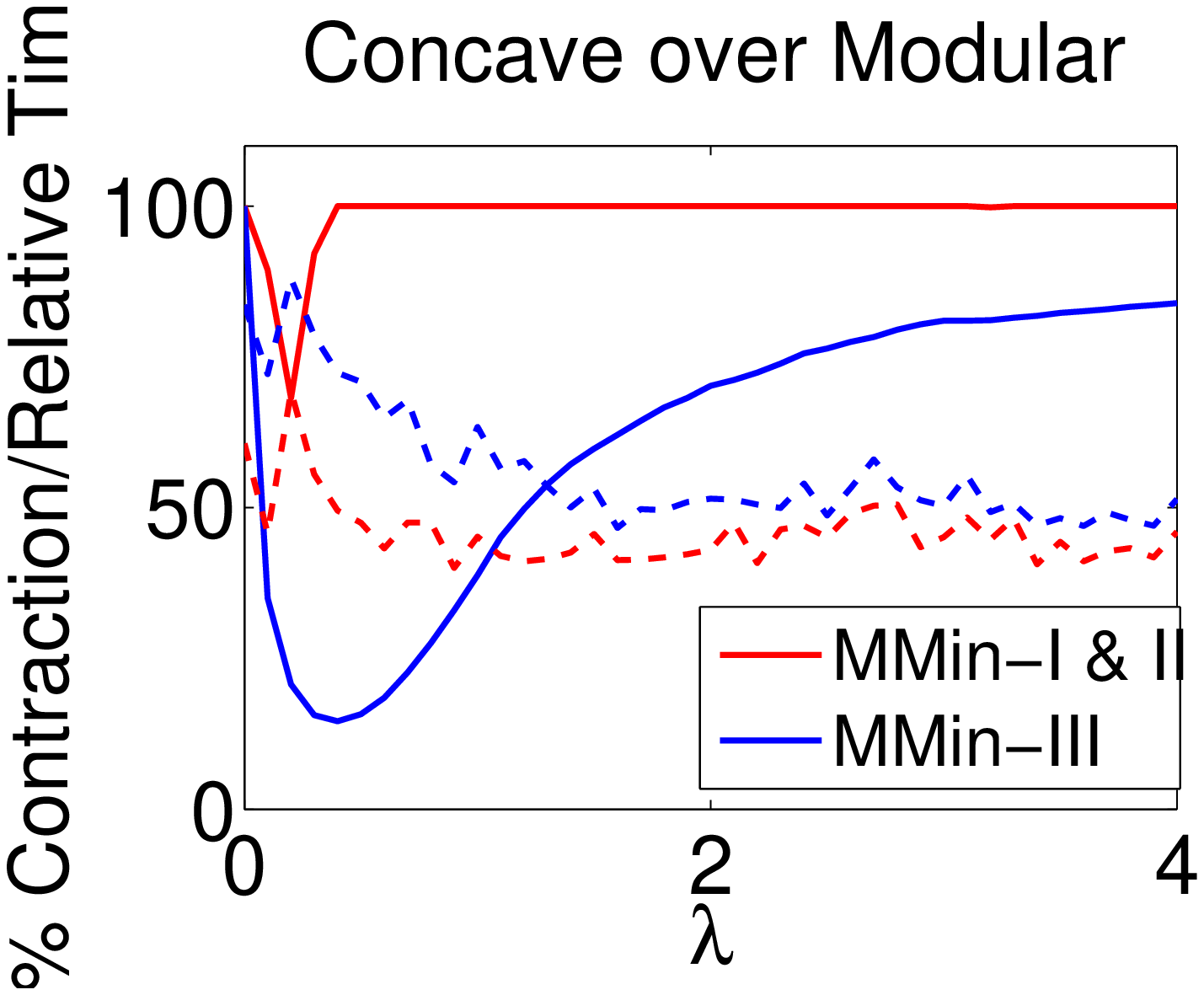} \hspace{-5pt} 
  \caption{Lattice reduction (solid line), and runtime (\%) of \mmin+min-norm relative to unadorned min-norm (dotted).\looseness-1 
}
  \label{fig:unconsmin}
\end{figure}
\begin{figure*}[t]
  \centering
\hspace{-10pt}
\includegraphics[width=0.23\textwidth]{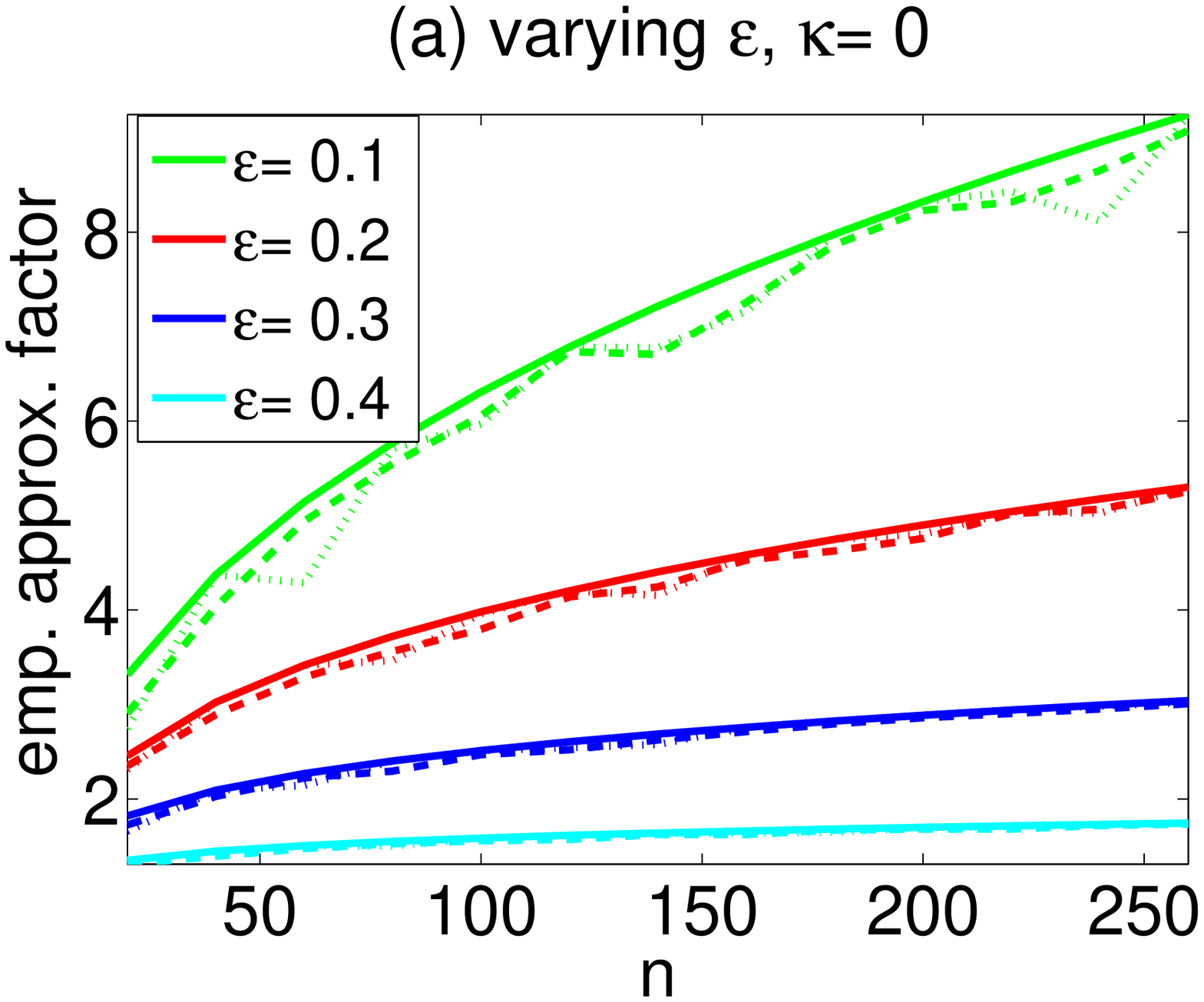}\hspace{-10pt} 
  ~ 
\includegraphics[width=0.23\textwidth,height=80pt]{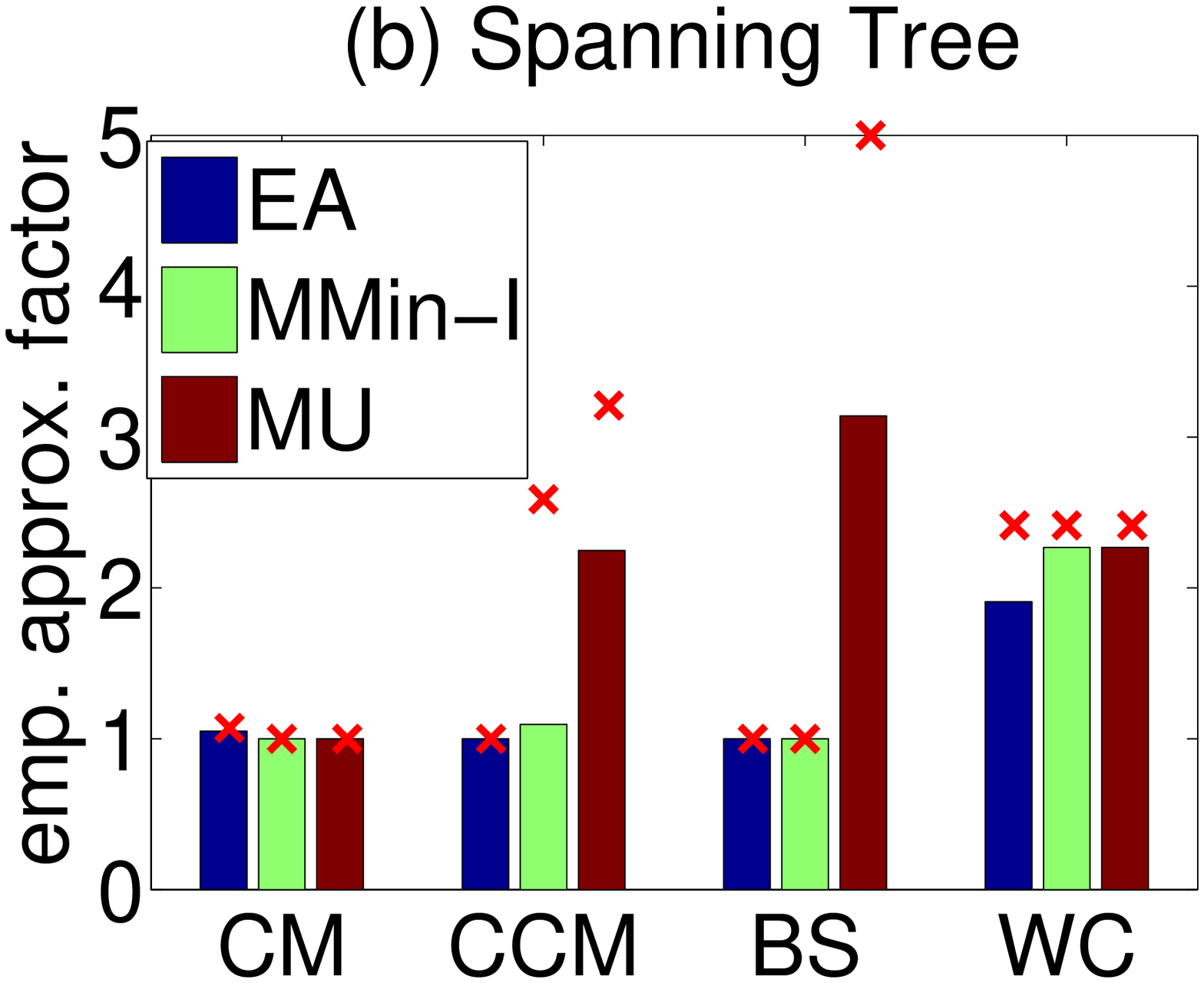} \hspace{-18pt} 
  ~ 
\includegraphics[width=0.23\textwidth,height=80pt]{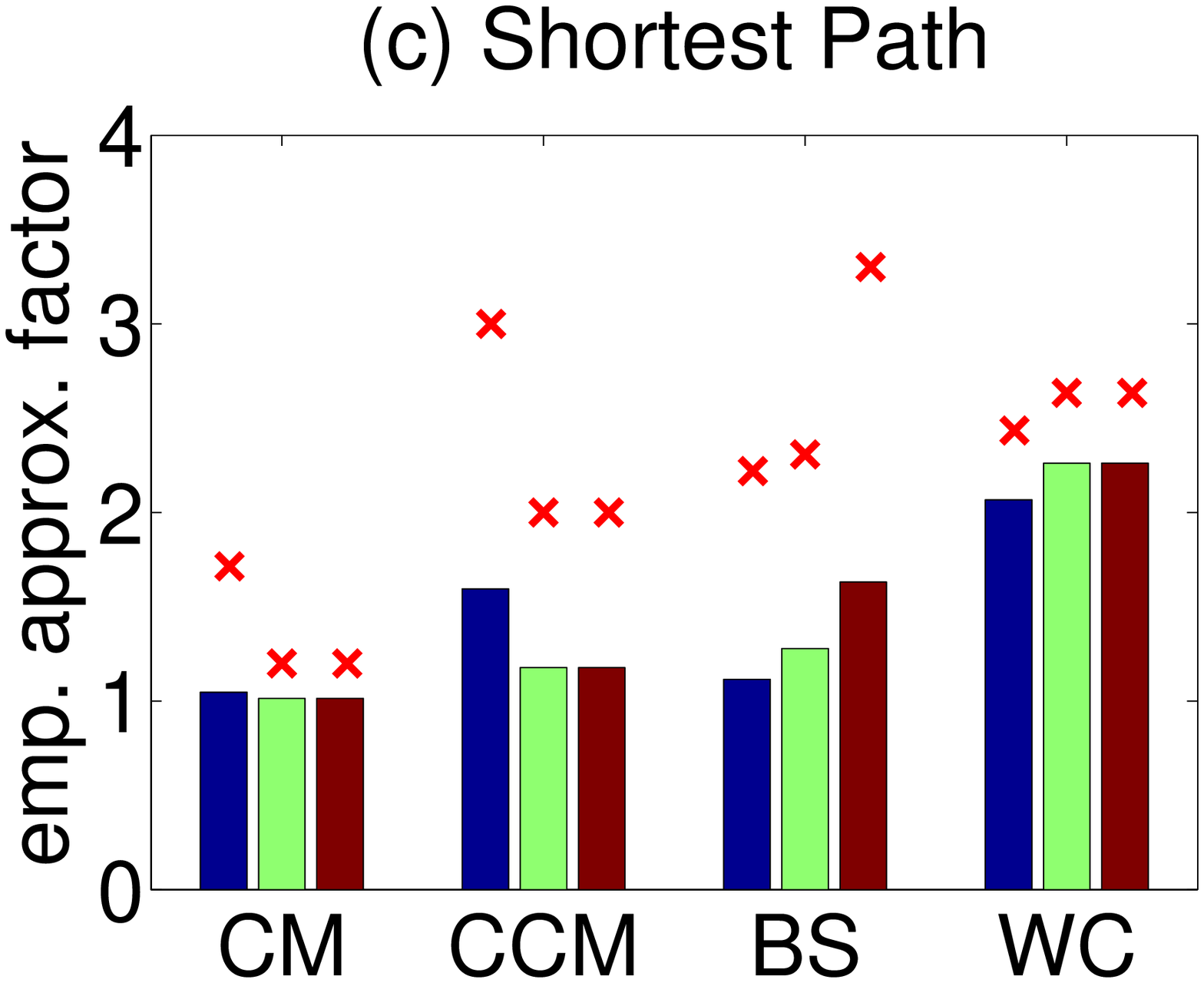}\hspace{-18pt} 
  ~
\includegraphics[width=0.23\textwidth,height=80pt]{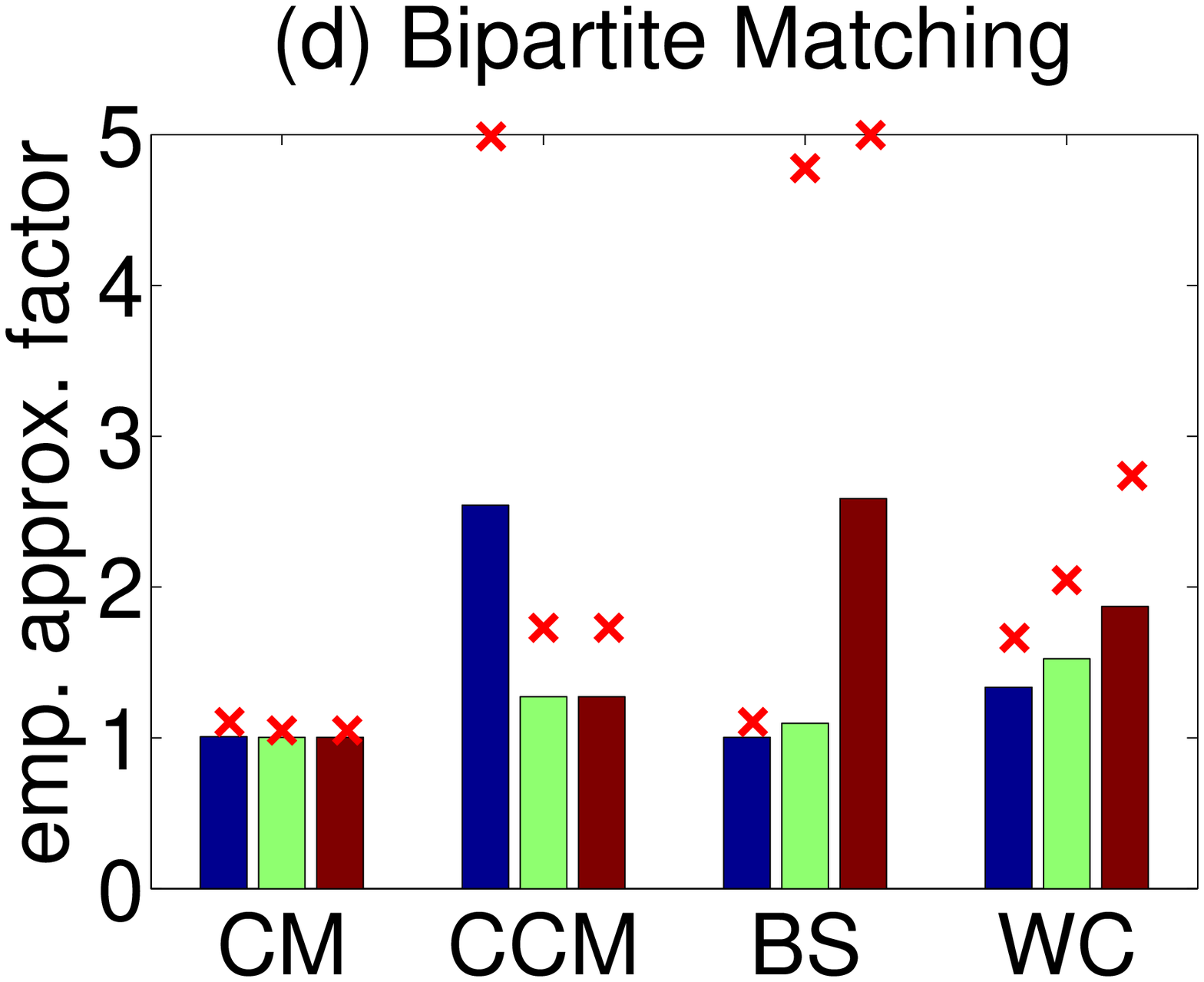}\hspace{-10pt} 
  \caption{Constrained minimization for worst-case (a) 
  and average-case (b-d) instances. In (a), Dashed lines: \mmin{}, dotted lines: EA, solid lines: theoretical bound. 
In (b - d), bars are average approximation factors and crosses worst observed results. CM - Concave over Mod., CCM - Clust. Concave Mod., BS - Best Set and WC - Worst Case
  }
  \label{fig:consmin}
\end{figure*}

\JTR{For Rishabh: fix the .eps files, as they have the (c), (d), (e)
  embedded in the figure, and the figs are currently labeled (a), (c),
  (d), (e).}

\paragraph{Constrained minimization.} For constrained minimization, we compare \mmini-I to two methods: a simple algorithm (MU) that minimizes the upper bound $g(X) = \sum_{i\in X}f(i)$ \cite{goel2009approximability} (this is identical to the first iteration of \mmini-I), and a more complex algorithm (EA) that computes an approximation to the submodular polyhedron~\cite{goemans2009approximating} and in many cases yields a theoretically optimal approximation. MU has the theoretical bounds of Theorem~\ref{SAAthm}, while EA achieves a worst-case approximation factor of $O(\sqrt{n} \log n)$. 
We show two experiments: the theoretical worst-case and average-case instances. Figure~\ref{fig:consmin} illustrates the results.

\paragraph{Worst case.} We use a very hard cost function \cite{goemans2009approximating} 
\begin{equation}
  \label{eq:3}
  f(X) = \min\{|X|, |X \cap \bar{R}| + \beta, \alpha\},
\end{equation}
where $\alpha = n^{1/2 + \epsilon}$ and $\beta = n^{2\epsilon}$, and $R$ is a random set such that $|R| = \alpha$. This function is the theoretical worst case. Figure~\ref{fig:consmin}~shows results for cardinality lower bound constraints; the results for other, more complex constraints are similar.
As $\epsilon$ shrinks, the problem becomes harder. In this case, EA and \mmin-I achieve about the same empirical approximation factors, which matches the theoretical guarantee of $n^{1/2 - \epsilon}$.
%

\paragraph{Average case.} We next compare the algorithms on more
realistic functions that occur in
applications. Figure~\ref{fig:consmin} shows the empirical
approximation factors for minimum submodular-cost spanning tree,
bipartite matching, and shortest path. We use four classes of
randomized test functions: (1) concave (square root or log) over
modular (CM), (2) clustered CM (CCM) of the form $f(X) = \sum_{i =
  1}^k \sqrt{w(X \cap C_k)}$ for clusters $C_1, \cdots, C_k$, (3) Best
Set (BS) functions where the optimal feasible set $R$ is chosen
randomly 
($f(X) = I(|X \cap R| \geq 1) + \sum_{j \in R \backslash X} w_j$) 
\JTR{I changed this from $R \backslash A$ to $R \backslash X$, double check.}\RTJ{thats right}
and (4) worst case-like functions (WC) similar to
equation~\eqref{eq:3}. Functions of type (1) and (2) have been used in
speech and computer vision~\cite{linacl, jegelka2011-nonsubmod-vision,
  rkiyeruai2012} and have reduced curvature ($\kappa_f <
1$). Functions of type (3) and (4) have $\kappa_f = 1$. In all four
cases, we consider both sparse and dense graphs, with random weight
vectors $w$. The plots show averages over $20$ instances of these
graphs. \arxivalt{For sparse graphs, we consider grid like graphs in
  the form of square grids, grids with diagonals and cubic grids. For
  dense graphs, we sparsely connect a few dense cluster subgraphs.
For matchings, we restrict ourselves to bipartite graphs, and consider both sparse and dense variants of these.
}{For more details, please refer to~\cite{ouricml2013supp}.}

First, we observe that in many cases, \mmini{} clearly outperforms MU. This suggests the practical utility of more than one iteration. Second, despite its simplicity, \mmini{} performs comparably to EA, and sometimes even better. 
In summary, the experiments suggest that the complex EA only gains on a few worst-case instances, whereas in many (average) cases, \mmini{} yields near-optimal results (factor 1--2).
In terms of running time, \mmini{} is definitely preferable: on small instances (for example $n = 40$), our Matlab implementation of \mmini{} takes 0.2 seconds, while EA needs about 58 seconds. On larger instances ($n=500$), the running times differ on the order of seconds versus hours.


\arxiv{\begin{table*}
\begin{center}
\small{
    \begin{tabular}{ | c | c | c | c |}
    \hline
    Constraints & Subgradient & Approximation bound  & Lower bound\\ \hline
    Unconstrained  & Random Permutation (RP) & $1/4$ & $1/2$ \\
    Unconstrained  & Random Adaptive (RA) & $1/4$ & $1/2$ \\
    Unconstrained  & Randomized Local Search (RLS) & $1/3- \eta$ & $1/2$ \\
    Unconstrained  & Deterministic Local Search (DLS) & $1/3 - \eta$ & $1/2$ \\
    Unconstrained  & Bi-directional greedy  (BG)& $1/3$ & $1/2$ \\
    Unconstrained  & Randomized Bi-directional greedy  (RG)& $1/3$ & $1/2$ \\
    Cardinality  & Greedy & $\frac{1}{\kappa_f}(1 - e^{-\kappa_f})$ & $\frac{1}{\kappa_f}(1 - e^{-\kappa_f})$ \\
   Matroid  & Greedy & $1/(1 + \kappa_f)$ & $\frac{1}{\kappa_f}(1 - e^{-\kappa_f})$ \\
   Knapsack  & Greedy & $1 - 1/e$ & $1 - 1/e$ 
 \\\hline
    \end{tabular}
 \caption{Approximation factors obtained through specific subgradients for submodular maximization (see text for details).\label{tab:2}}}
 \end{center}
 \end{table*}}

\section{Submodular maximization}\label{approxmax}
Just like for minimization, for submodular maximization too we obtain a family of algorithms where each member is specified by a distinct schedule of subgradients. We will only select subgradients that are vertices of the subdifferential, i.e., each subgradient corresponds to a permutation of $V$. For any of those choices, \mmax{} converges quickly.
To bound the running time, we assume that we proceed only if we make sufficient progress, i.e.,
if $f(X^{t+1}) \geq (1 + \eta)f(X^t)$.
\begin{lemma} \label{maxcomp}
\mmax{} with $X^0 = \argmax_j f(j)$ runs in time 
$O(T \log_{1+\eta} n)$, where $T$ is the time for maximizing a modular function subject to $X \in \mathcal C$.
\end{lemma} 
\arxiv{
\begin{proof}
Let $X^*$ be the optimal solution, then
\begin{equation}
  f(X^*) \leq \sum_{i \in X^*} f(j) \leq n \max_{j \in V} f(j) = n f(X^0).\label{eq:6}
\end{equation}
Furthermore, we know that $f(X^t) \geq (1+\eta)^tf(X^0)$. Therefore, we have reached the maximum function value after at most $(\log n) / \log(1+\eta)$ iterations. 
%
\end{proof}
}
In practice, we observe that \mmax{} terminates within 3-10 iterations. We next consider specific subgradients and their theoretical implications. For unconstrained problems, we assume the submodular function to be non-monotone (the results trivially hold for monotone functions too); for constrained problems, we assume the function $f$ to 
be monotone nondecreasing.
Our results rely on the observation that many maximization algorithms actually compute a specific subgradient and run \mmax{} with this subgradient. To our knowledge, this observation is new.
\notarxiv{The proofs for the statements in Section~\ref{sec:unconstr_max} may be found in \cite{ouricml2013supp}.}

\subsection{Unconstrained Maximization}\label{sec:unconstr_max}
\paragraph{Random Permutation (RA/RP).} In iteration $t$, we randomly pick a permutation $\sigma$ that defines a subgradient at $X^{t-1}$, i.e., $X^{t-1}$ is assigned to the first $|X^{t-1}|$ positions.
At $X^0 = \emptyset$, this can be any permutation. Stopping after the first iteration (RP) achieves an approximation factor of $1/4$ in expectation, and $1/2$ for symmetric functions. Making further iterations (RA) only improves the solution.
%

\arxiv{
\begin{lemma}\label{randthm}
  When running Algorithm RP with $X^0 = \emptyset$, it holds after one iteration that 
  %
  $\mathbf{E}(f(X^1)) \geq \frac{1}{4} f(X^*)$ if $f$ is a
  general non-negative submodular function, and $\mathbf{E}(f(X^1))
  \geq \frac{1}{2} f(X^*)$ if $f$ is symmetric.
\end{lemma}%
\begin{proof}
Each permutation has the same probability $1/n!$ of being chosen. Therefore, it holds that
\begin{align}
\mathbf{E}(f(X^1)) &= \mathbf{E}_{\sigma}(\max_{X \subseteq V} h^{\sigma}_{\emptyset}(X)) \\
\label{eq:bd1}
				&= \frac{1}{n!} \sum_{\sigma} \max_{X \subseteq V} h^{\sigma}_{\emptyset}(X) 
\end{align}
Let $\emptyset \subseteq S^{\sigma}_1 \subseteq S^{\sigma}_2 \cdots S^{\sigma}_n = V$ be the chain corresponding to a given permutation $\sigma$. We can bound
\begin{equation}\label{eq:bd2}
\max_{X \subseteq V} h^{\sigma}_{\emptyset}(X) \geq \sum_{k = 0}^n \frac{\binom{n}{k}}{2^n} f(S^{\sigma}_k)
\end{equation}
because $\max_{X \subseteq V} h^{\sigma}_{\emptyset}(X) \geq f(S^{\sigma}_k), \forall k$ and $\sum_{k = 0}^n \frac{\binom{n}{k}}{2^n} = 1$. Together, Equations~\eqref{eq:bd1} and \eqref{eq:bd2} imply that
\begin{align}
\mathbf{E}(f(X^1)) &\geq \mathbf{E}_{\sigma}(\max_{X \subseteq V} h^{\sigma}_{\emptyset}(X)) \\
				&= \sum_{\sigma}\sum_{k = 0}^n \frac{\binom{n}{k}}{2^n} f(S^{\sigma}_k) \frac{1}{n!} \\
				&= \sum_{k = 0}^n \frac{\binom{n}{k}}{n! 2^n} \sum_{\sigma} f(S^{\sigma}_k) \\
				&= \sum_{k = 0}^n \frac{\binom{n}{k}}{n! 2^n} k! (n - k)! \sum_{S: |S| = k} f(S) \\
				&= \sum_{S} \frac{f(S)}{2^n} \\
				&= \mathbf{E}_S(f(S))
\end{align}
By $\mathbf{E}_S(f(S))$, we denote the expected function value when the set $S$ is sampled uniformly at random, i.e., each element is included with probability $1/2$. \cite{janvondrak} shows that $\mathbf{E}_S(f(S)) \geq \frac{1}{4}f(X^*)$.
For symmetric submodular functions, the factor is $\frac{1}{2}$.
\end{proof}

}
\paragraph{Randomized local search (RLS).} 
Instead of using a completely random subgradient as in RA, we fix the positions of two elements:
the permutation must satisfy 
that $\sigma^t(|X^t| + 1) \in \argmax_j f(j | X^t)$ and $\sigma^t(|X^t| - 1) \in \argmin_j f(j | X^t \backslash j)$. The remaining positions are assigned randomly.
An $\eta$-approximate version of \mmax{} with such subgradients
returns an $\eta$-approximate local maximum that achieves an improved approximation factor of $1/3 - \eta$ in $O(\frac{n^2\log n}{\eta}$) iterations. 

\arxiv{\begin{lemma}\label{locrandopt}
Algorithm RLS returns a local maximum $X$ that satisfies $\max\{f(X), f(V \backslash X)\} \geq (\frac{1}{3} - \eta) f(X^*)$ in $O(\frac{n^2\log n}{\eta}$) iterations.
%
\end{lemma}
\begin{proof}
At termination ($t=T$), it holds that $\max_j f(j | X^T) \leq 0$ and $\min_j f(j | X^T \setminus j) \geq 0$; this implies that the set $X^t$ is local optimum.

To show local optimality, recall that the subgradient $h^{\sigma^T}_{X^T}$ satisfies $h^{\sigma^T}_{X^T}(X^T) = f(X^T)$, and $h^{\sigma^T}_{X^T}(Y) \geq h^{\sigma^T}_{X^T}(X^T)$ for all $Y \subseteq V$. Therefore, it must hold that $max_{j \notin X^T} f(j | X^T) = \max_{j \notin X^T} h^{\sigma^T}_{X^T}(j) \leq 0$, and $\min_{j \in X^T} f(j | X^T \backslash j) =   h^{\sigma^T}_{X^T}(j) \geq 0$, which implies that the set $X^T$ is a local maximum.

We now use a result by~\cite{janvondrak} showing that if a set $X$ is a local optimum, then $f(X) \geq \frac{1}{3} f(X^*)$ if $f$ is a general non-negative submodular set function and $f(X) \geq \frac{1}{2} f(X^*)$ if $f$ is a symmetric submodular function. 
If the set is an $\eta$-approximate local optimum, we obtain a $\frac{1}{3} - \eta$ approximation~\cite{janvondrak}. A complexity analysis similar to Theorem~\ref{maxcomp} reveals that the worst case complexity of this algorithm is $O(\frac{n^2\log n}{\eta})$.
\end{proof}

Note that even finding an exact local maximum is hard for submodular functions~\cite{janvondrak}, and therefore it is necessary to resort to an $\eta$-approximate version, which converges to an $\eta$-approximate local maximum.
}

\paragraph{Deterministic local search (DLS).} A completely deterministic variant of RLS defines the permutation by an entirely greedy ordering. We define permutation $\sigma^t$ used in iteration $t$ via the chain $\emptyset=S^{\sigma^t}_0 \subset S^{\sigma^t}_1 \subset \ldots \subset S^{\sigma^t}_n$ it will generate. The initial permutation is $\sigma^0(j) = \argmax_{k \notin S^{\sigma^0}_{j-1}} f(k | S^{\sigma^0}_{j-1})$ for $j = 1, 2, \ldots$. In subsequent iterations $t$, the permutation $\sigma^t$ is
\begin{align}
\label{eq:locsearch}
\sigma^t(j) = 
\begin{cases}
\sigma^{t-1}(j) & \text{ if } t \text{ even, } j \in X^{t-1}\\
\argmax_k f(k | S^{\sigma^t}_{j-1}) & \text{ if } t \text{ even, } j \notin X^{t-1}\\
\argmin_k f(k | S^{\sigma^t}_{j+1}\backslash k) & \text{ if } t \text{ odd, } j \in X^{t-1} \\
\sigma^{t-1}(j) & \text{ if } t \text{ odd, } j \notin X^{t-1}.
\end{cases} \nonumber
\end{align}
This schedule is equivalent to the deterministic local search (DLS) algorithm by~\cite{janvondrak}, and therefore achieves an approximation factor of $1/3 - \eta$. 

\paragraph{Bi-directional greedy (BG).} 
The procedures above indicate that greedy and local search algorithms implicitly define specific chains and thereby subgradients. Likewise, the deterministic bi-directional greedy algorithm by \cite{feldman2012optimal} induces a distinct permutation of the ground set. It is therefore equivalent to \mmax{} with the corresponding subgradients and achieves an approximation factor of $1/3$. 
This factor improves that of the local search techniques by removing $\eta$.
Moreover, unlike for local search, the $1/3$ approximation holds already after the first iteration.

\arxiv{

\begin{lemma}\label{lem:bg}
The set $X^1$ obtained by Algorithm~\ref{mirroropt} with the subgradient equivalent to BG satisfies that
$f(X) \geq \frac{1}{3} f(X^*)$. 
\end{lemma}
\begin{proof}
Given an initial ordering $\tau$, the bi-directional greedy algorithm by \cite{feldman2012optimal} generates a chain of sets. Let $\sigma^\tau$ denote the permutation defined by this chain, obtainable by mimicking the algorithm. We run \mmax{} with the corresponding subgradient.
By construction, the set $S^\tau$ returned by the bi-directional greedy algorithm is contained in the chain. Therefore, it holds that
%
\begin{align}
  f(X^1) &\geq \max_{X \subseteq V} h^{\sigma^{\tau}}_{\emptyset}(X) \\
  &\geq \max_k f(S^{\sigma^{\tau}}_k) \\
  &\geq f(S^{\tau}) \\
  &\geq \frac{1}{3} f(X^*). 
\end{align}
The first inequality follows since the subgradient is tight for all sets in the chain. For the second inequality, we used that
$S^{\tau}$ belongs to the chain, and hence $S^{\tau} = S^{\sigma^{\tau}}_j$ for some $j$. The last inequality follows from the approximation factor satisfied by $S^\tau$~\cite{feldman2012optimal}. We can continue the algorithm, using any one of the adaptive schedules above to get a locally optimal solution. This can only improve the solution.
\end{proof}
}

%
%
\paragraph{Randomized bi-directional greedy (RG).} 
Like its deterministic variant, the randomized bi-directional greedy algorithm by
\cite{feldman2012optimal} can be shown to run \mmax{} with a specific subgradient. 
Starting from $\emptyset$ and $V$, it implicitly defines a random chain of subsets and thereby (random) subgradients. 
%
A simple analysis shows that this subgradient leads to the best possible approximation factor of $1/2$ in expectation. %
 
 \arxiv{Like its deterministic counterpart, the Randomized bi-directional Greedy algorithm (RG) by \cite{feldman2012optimal} induces a (random) permutation $\sigma^{\tau}$ 
 based on an initial ordering $\tau$. 

\begin{lemma}
If the subgradient in \mmax{} is determined by $\sigma^{\tau}$, then the set $X^1$ after the first iteration satisfies $\mathbf{E}(f(X^1)) \geq \frac{1}{2} f(X^*)$, where the expectation is taken over the randomness in $\sigma^{\tau}$.
 \end{lemma} 
 \begin{proof}
The permutation $\sigma^\tau$ is obtained by a randomized algorithm, but once $\sigma^\tau$ is fixed, the remainder of \mmax{} is deterministic.
By an argumentation similar to that in the proof of Lemma~\ref{lem:bg}, it holds that
\begin{align}
\mathbf{E}(f(X)) &\geq \mathbf{E}( \max_X h^{\sigma^{\tau}}_{\emptyset}(X)) \\
				&\geq \mathbf{E}(\max_k f(S^{\sigma^{\tau}}_k)) \\
				&\geq \mathbf{E}( f(S^{\sigma^{\tau}})) \\
				&\geq \frac{1}{2} f(X^*)
\end{align}
The last inequality follows from a result in~\cite{feldman2012optimal}.
\end{proof}
 }
 \subsection{Constrained Maximization}
 In this final section, we analyze subgradients for maximization subject to the constraint $X \in \mathcal{C}$. Here we assume that $f$ is monotone. 
An important subgradient results from the greedy permutation $\sigma^g$, defined as
 \begin{align}
 \sigma^g(i) \in \argmax_{j \notin S^{\sigma^g}_{i-1}\text{ and } S^{\sigma^g}_{i-1} \cup \{j\} \in \mathcal C } f(j | S^{\sigma^g}_{i-1}).
 \end{align} 
This definition might be partial; we arrange any remaining elements arbitrarily.
When using the corresponding subgradient $h^{\sigma^g}$, 
we recover a number of approximation results already after one iteration:
%
\begin{lemma}\label{thm:greedy}
Using $h^{\sigma^g}$ in iteration 1 of \mmax{} yields the following approximation bounds for
$X^1$:
\addtolength{\leftmargini}{-1em}
\begin{itemize}\setlength{\itemsep}{0pt}
\item $\frac{1}{\kappa_f}(1 - e^{-\kappa_f})$, if $\mathcal C = \{X \subseteq V: |X| \leq k\}$ 
\item $\frac{1}{p + \kappa_f}$, for the intersection $\mathcal C\! = \! \cap_{i = 1}^p \mathcal I_i$ of $p$ matroids 
\item $\frac{1}{\kappa_f} (1 - (\frac{K - \kappa_f}{K})^k)$, for any down-monotone constraint $\mathcal C$, where $K$ and $k$ are the maximum and minimum cardinality of the maximal feasible sets in $\mathcal C$.\looseness-1
%
\end{itemize}
 \end{lemma} 
 \arxiv{
 \begin{proof}
We prove the first result for cardinality constraints. The proofs for the matroid and general down-monotone constraints are analogous.
By the construction of $\sigma^g$, the set $S^{\sigma^g}_k$ is exactly the set returned by
the greedy algorithm. This implies that
\begin{align}
  f(X^1) &\geq \argmax_{X: |X| \leq k} h^{\sigma^g}_{\emptyset}(X) \\
  &\geq  h^{\sigma^g}_{\emptyset}(S^{\sigma^g}_k) \\
  &= f(S^{\sigma^g}_k)  \\
  &\geq \frac{(1 - e^{-\kappa_f})}{\kappa_f} f(X^*).
\end{align}
The last inequality follows from~\cite{nemhauser1978, conforti1984submodular}. 
\end{proof}

A very similar construction of a greedy permutation provides bounds for budget constraints, i.e.,
$c(S) \triangleq \sum_{i \in S}c(i) \leq B$ for some given nonnegative costs $c$. In particular, define a permutation as:
 \begin{align}
 \sigma^g(i) \in \argmax_{j \notin S^{\sigma^g}_{i-1}, c(S^{\sigma^g}_{i-1} \cup \{j\}) \leq B} \frac{f(j | S^{\sigma^g}_{i-1})}{c(j)}.
 \end{align} 

%
The following result then follows from~\cite{linbudget,sviridenko2004note}.
\begin{lemma}\label{greedybudget}
Using $\sigma^g$ in \mmax{} under the budget constraints yields: 
\begin{equation}
  \label{eq:8}
  \max\{\max_{i: c(i) \leq B} f(i), f(X^1)\} \geq (1 - 1/\sqrt{e})f(X^*).
\end{equation}
Let $\sigma^{ijk}$ be a permutation with $i,j,k$ in the first three positions, and the remaining arrangement greedy. Running $O(n^3)$ restarts of MM yields sets $X_{ijk}$ (after one iteration) with
\begin{equation}
  \label{eq:9}
  \max_{i,j,k \in V} f(X_{ijk}) \,\geq (1-1/e) f(X^*).
\end{equation}
\end{lemma}
The proof is analogous to that of Lemma~\ref{thm:greedy}.
}
\notarxiv{
A similar result holds for 
Knapsack constraints. 
The proof of Lemma~\ref{thm:greedy} \cite{ouricml2013supp} relies on the observation that the maximizer of
the function $m_h$ for the subgradient $h = h^{\sigma^g}$ is never worse than the result of a greedy algorithm. The bounds follow from~\cite{conforti1984submodular}. 
} 
\arxiv{Table~\ref{tab:2}
 lists results for monotone submodular maximization under different constraints. 

It would be interesting if some of the constrained variants of non-monotone submodular maximization could be naturally subsumed in our framework too. In particular, some recent algorithms~\cite{lee2009non, matroidimproved} propose local search based techniques to obtain constant factor approximations for non-monotone submodular maximization under knapsack and matroid constraints. Unfortunately, these algorithms require swap operations along with inserting and deleting elements.
We do not currently know how to phrase these swap operations via our framework and leave this relation as an open problem.

While a number of algorithms cannot be naturally seen as an instance of our framework, we show in the following section that any polynomial time approximation algorithm for unconstrained or constrained variants of submodular optimization can be ultimately seen as an instance of our algorithm, via a polynomial-time computable subgradient.
}
 \subsection{Generality}
The correspondences between \mmax{} and maximization algorithms hold even more generally:
\begin{theorem}\label{anyalgmax}
For any polynomial-time unconstrained submodular maximization algorithm that achieves an approximation factor $\alpha$, there exists a schedule of subgradients (obtainable in polynomial time) that, if used within \mmax, leads to a solution with the same approximation factor $\alpha$.
\end{theorem}
\arxiv{
The proof relies on the following observation.
\begin{lemma}\label{permeqmax}
Any submodular function $f$ satisfies
\begin{equation}\label{permeqmaxeqn}
\max_{X \in \mathcal C} f(X) = \max_{X \in \mathcal C, h \in \mathcal P_f} h(X) =\max_{X \in \mathcal C, \sigma \in \Sigma} h^{\sigma}_{\emptyset}(X).
\end{equation}
 \end{lemma} 
Lemma~\ref{permeqmax} implies that there exists a permutation (and equivalent subgradient) with which \mmax{} finds the optimal solution in the first iteration. Known hardness results~\cite{feige1998threshold} imply that this permutation may not be obtainable in polynomial time.
\begin{proof}{\emph{(Lemma~\ref{permeqmax})}}
The first equality in Lemma~\ref{permeqmax} follows from the fact that any submodular function $f$ can be written as
\begin{equation}
  \label{eq:10}
  f(X) = \max_{h \in P_f}h(X). 
\end{equation}
For the second equality, we use the fact that a linear program over a polytope has a solution at one of the extreme points of the corresponding polytope.
%
\end{proof}
We can now prove Theorem~\ref{anyalgmax}
\begin{proof}{\emph{(Thm.~\ref{anyalgmax})}}
Let $Y$ be the set returned by the approximation algorithm; this set is polynomial-time computable by definition. Let $\tau$ be an arbitrary permutation that places the elements in $Y$ in the first $|Y|$ positions. The subgradient $h^\tau$ defined by $\tau$ is a subgradient both for $\emptyset$ and for $Y$. Therefore, using $X^0 = \emptyset$ and $h^\tau$ in the first iteration, we obtain a set $X^1$ with
\begin{align}\label{eq:subgrad_optim}
  f(X^1) \geq h^{\tau}_{\emptyset}(X^1) \geq h^{\tau}_{\emptyset}(Y) = f(Y) \geq \alpha f(X^*).
\end{align}
%
%
The equality follows from the fact that $Y$ belongs to the chain of $\tau$. 
\end{proof}
While the above theorem shows the optimality of \mmax{} in the unconstrained setting, a similar result holds for the constrained case:
\begin{corollary}\label{cor:constrained}
Let $\mathcal{C}$ be any constraint such that a linear function can be exactly maximized over $\mathcal{C}$. For any polynomial-time algorithm for submodular maximization over $\mathcal{C}$ that achieves an approximation factor $\alpha$, there exists a schedule of subgradients (obtainable in polynomial time) that, if used within \mmax{}, leads to a solution with the same approximation factor $\alpha$.
\end{corollary}
The proof of Corollary~\ref{cor:constrained} follows directly from the Theorem~\ref{anyalgmax}. 
}
\notarxiv{Under mild assumptions, Theorem~\ref{anyalgmax} holds even for constrained maximization.}
Lastly, we pose the question of selecting the optimal subgradient in each iteration. An optimal subgradient $h$ would lead to a function $m_h$ whose maximization yields the largest improvement. 
Unfortunately, obtaining such an ``optimal'' subgradient is impossible:
\begin{theorem}
  The problem of finding the optimal subgradient $\sigma^{OPT} = \argmax_{\sigma, X \subseteq V} h^{\sigma}_{X^t}(X)$ in Step 4 of Algorithm~\ref{mirroropt} is NP-hard even when $\mathcal C = 2^V$. Given such an oracle, however, \mmax{} using subgradient $\sigma^{OPT}$ returns a global optimizer.
\end{theorem}\looseness-1
\arxiv{
\begin{proof}
Lemma~\ref{permeqmax} implies that an optimal subgradient at
 $X^0=\emptyset$ or $X^0=V$ is a subgradient at an optimal solution. An argumentation as in Equation~\eqref{eq:subgrad_optim} shows that using this subgradient in MM leads to an optimal solution. Since this would solve submodular maximization (which is NP-hard), it must be NP-hard to find such a subgradient.

To show that this holds for arbitrary $X^t$ (and correspondingly at every iteration), we use that the submodular subdifferential can be expressed as a direct product between a submodular polyhedron and an anti-submodular polyhedron~\cite{fujishige2005submodular}. Any problem involving an optimization over the sub-differential, can then be expressed as an optimization over a submodular polyhedron (which is a subdifferential at the empty set) and an anti-submodular polyhedron (which is a subdifferential at $V$)~\cite{fujishige2005submodular}. Correspondingly, Equation~\eqref{permeqmaxeqn} can be expressed as the sum of two submodular maximization problems. 
\STR{Commented out the rest; reducing sth \emph{to} submodular minimization does not prove that it is NP-hard!}
\end{proof}
}

\subsection{Experiments}
We now empirically compare variants of \mmax{} with different subgradients.
As a test function, we use the objective of \cite{lin2009select}, 
$f(X) = \sum_{i \in V} \sum_{j \in X} s_{ij} - \lambda \sum_{i, j \in X} s_{ij}$, where $\lambda$  is a redundancy parameter. This non-monotone function was used to find 
the most diverse yet relevant subset of objects\arxiv{ in a large corpus}.
We use the objective with both synthetic and real data. 
We generate $10$ instances of random similarity matrices $\{s_{ij}\}_{ij}$ and 
vary $\lambda$ from $0.5$ to 1.
Our real-world data is the Speech Training data subset selection problem~\cite{lin2009select} on the TIMIT corpus~\cite{timit}, using the string kernel metric~\cite{rousu2006efficient} for similarity. 
We use $20 \leq n \leq 30$ so that the exact solution can still be computed with the algorithm of \cite{goldengorin1999maximization}. \looseness-1

We compare the algorithms DLS, BG, RG, RLS, RA and RP, and a baseline
RS that picks a set uniformly at random. RS achieves a $1/4$
approximation in expectation~\cite{janvondrak}.  For random
algorithms, we select the best solution out of 5 repetitions.
Figure~\ref{fig:max} shows that DLS, BG, RG and RLS dominate.  Even though RG
has the best theoretical worst-case bounds, it performs slightly
poorer than the local search ones and BG.
Moreover, \mmax{} with random subgradients (RP) is much better than choosing a set uniformly at random (RS). In general, the empirical approximation factors are much better than the theoretical worst-case bounds.
Importantly, the \mmax{} variants are extremely fast, about 200-500 times faster than the exact branch and bound technique of~\cite{goldengorin1999maximization}.


\section{Discussion and Conclusions}
In this paper, we introduced a general MM framework for submodular optimization algorithms. This framework is akin to the class of algorithms for minimizing the difference between submodular functions~\cite{narasimhanbilmes,rkiyeruai2012}. In addition, it may be viewed as a special case of a proximal minimization algorithm that uses Bregman divergences derived from submodular functions~\cite{iyermirrordescent}. To our knowledge this is the first generic and unifying framework of combinatorial algorithms for submodular optimization. \looseness-1

An alternative framework relies on relaxing the discrete optimization problem by using a continuous extension (the \lovasz{} extension for minimization and multilinear extension for maximization).
Relaxations have been applied to some constrained \cite{iwata2009submodular} and unconstrained~\cite{frbach1} minimization problems as well as maximization problems~\cite{feldman2012optimal}. Such relaxations, however, rely on a final rounding step that can be challenging --- the combinatorial framework obviates this step.
Moreover, our results show that in many cases, it yields good results very efficiently.
\begin{figure}[t]
  \centering
\hspace{-10pt}
\includegraphics[width=0.26\textwidth]{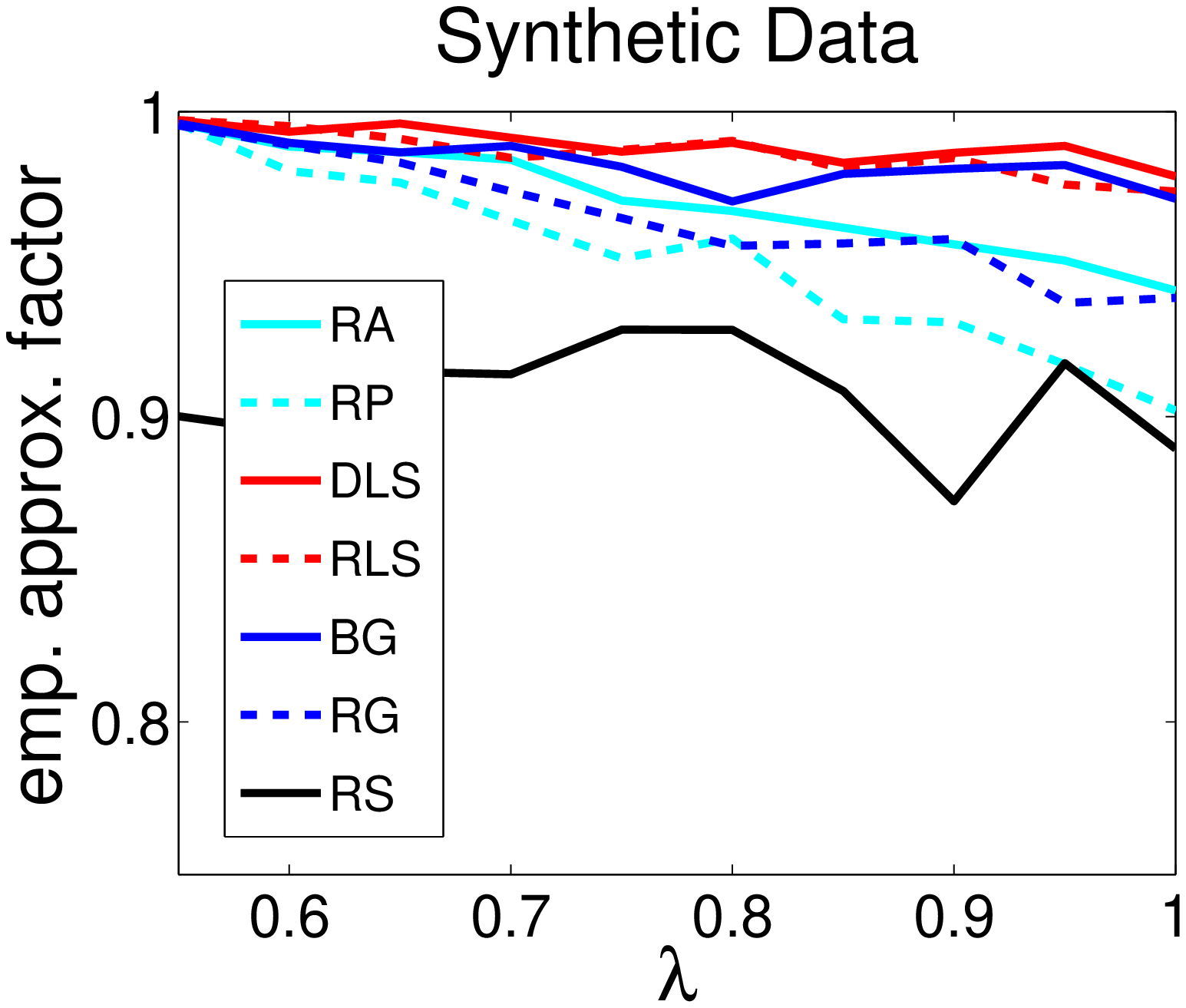}\hspace{-20pt}
  ~ 
\includegraphics[width=0.26\textwidth]{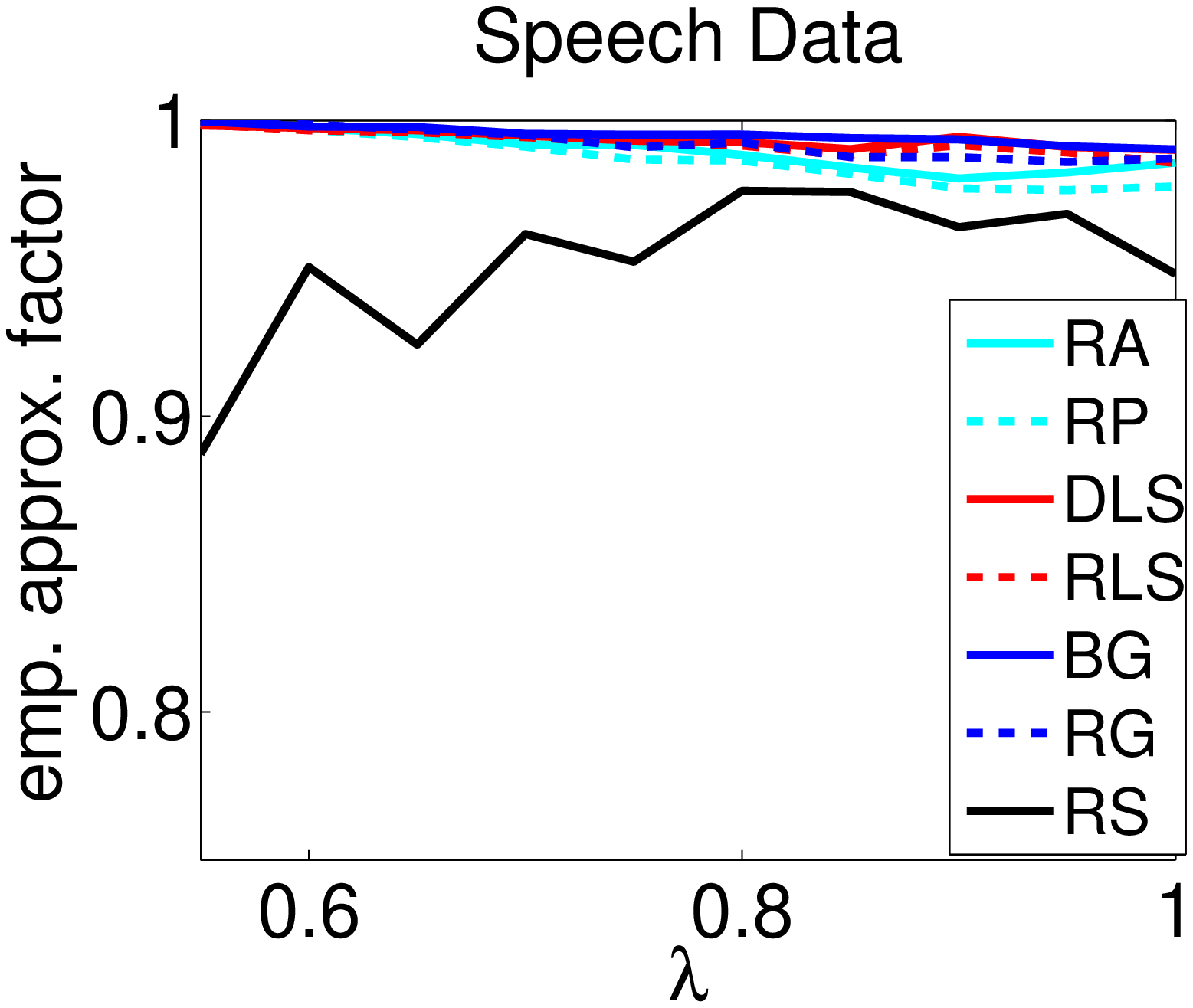} \hspace{-5pt}
  \caption{Empirical approximation factors for variants of \mmax. See Section~\ref{sec:unconstr_max} for legend details.
  }
  \label{fig:max}
\end{figure}

\looseness-1

{\bf Acknowledgments:} We thank Karthik Mohan, John Halloran and Kai Wei for
discussions. This material is based upon work supported by the National Science
Foundation under Grant No.\ IIS-1162606, and by a
Google, a Microsoft, and an Intel research award. This material is also based upon work supported in part by the Office of Naval Research under contract/grant number N00014-11-1-068, NSF CISE Expeditions award CCF-1139158 and DARPA XData Award FA8750-12-2-0331, and  gifts from Amazon Web Services, Google, SAP,  Cisco, Clearstory Data, Cloudera, Ericsson, Facebook, FitWave, General Electric, Hortonworks, Intel, Microsoft, NetApp, Oracle, Samsung, Splunk, VMware and Yahoo!.\looseness-1
\bibliographystyle{plain}
\bibliography{../Combined_Bib/submod}

\begin{thebibliography}{10}

\bibitem{averbakh94}
I.~Averbakh and O.~Berman.
\newblock Categorized bottleneck-minisum path problems on networks.
\newblock {\em Operations Research Letters}, 16:291--297, 1994.

\bibitem{frbach1}
F.~Bach.
\newblock {Learning with Submodular functions: A convex Optimization
  Perspective}.
\newblock {\em Arxiv}, 2011.

\bibitem{boykovJolly01}
Y.~Boykov and M.P. Jolly.
\newblock Interactive graph cuts for optimal boundary and region segmentation
  of objects in n-d images.
\newblock In {\em ICCV}, 2001.

\bibitem{feldman2012optimal}
N.~Buchbinder, M.~Feldman, J.~Naor, and R.~Schwartz.
\newblock A tight (1/2) linear-time approximation to unconstrained submodular
  maximization.
\newblock {\em In FOCS}, 2012.

\bibitem{conforti1984submodular}
M.~Conforti and G.~Cornuejols.
\newblock Submodular set functions, matroids and the greedy algorithm: tight
  worst-case bounds and some generalizations of the {R}ado-{E}dmonds theorem.
\newblock {\em Discrete Applied Mathematics}, 7(3):251--274, 1984.

\bibitem{delong2012minimizing}
A.~Delong, O.~Veksler, A.~Osokin, and Y.~Boykov.
\newblock Minimizing sparse high-order energies by submodular vertex-cover.
\newblock In {\em In NIPS}, 2012.

\bibitem{feige1998threshold}
U.~Feige.
\newblock A threshold of ln n for approximating set cover.
\newblock {\em Journal of the ACM (JACM)}, 1998.

\bibitem{janvondrak}
U.~Feige, V.~Mirrokni, and J.~Vondr{\'a}k.
\newblock Maximizing non-monotone submodular functions.
\newblock {\em SIAM J. COMPUT.}, 40(4):1133--1155, 2007.

\bibitem{fujishige2005submodular}
S.~Fujishige.
\newblock {\em Submodular functions and optimization}, volume~58.
\newblock Elsevier Science, 2005.

\bibitem{fujishige2011submodular}
S.~Fujishige and S.~Isotani.
\newblock A submodular function minimization algorithm based on the
  minimum-norm base.
\newblock {\em Pacific Journal of Optimization}, 7:3--17, 2011.

\bibitem{timit}
J.~Garofolo, Fisher Lamel, L., J.~W., Fiscus, D.~Pallet, and N.~Dahlgren.
\newblock Timit, acoustic-phonetic continuous speech corpus.
\newblock In {\em DARPA}, 1993.

\bibitem{goel2009approximability}
G.~Goel, C.~Karande, P.~Tripathi, and L.~Wang.
\newblock Approximability of combinatorial problems with multi-agent submodular
  cost functions.
\newblock In {\em FOCS}, 2009.

\bibitem{goemans2009approximating}
M.X. Goemans, N.J.A. Harvey, S.~Iwata, and V.~Mirrokni.
\newblock Approximating submodular functions everywhere.
\newblock In {\em SODA}, pages 535--544, 2009.

\bibitem{goldengorin1999maximization}
B.~Goldengorin, G.A. Tijssen, and M.~Tso.
\newblock {\em The maximization of submodular functions: Old and new proofs for
  the correctness of the dichotomy algorithm}.
\newblock University of Groningen, 1999.

\bibitem{hunter2004tutorial}
D.R. Hunter and K.~Lange.
\newblock A tutorial on {MM} algorithms.
\newblock {\em The American Statistician}, 2004.

\bibitem{iwata2009submodular}
S.~Iwata and K.~Nagano.
\newblock Submodular function minimization under covering constraints.
\newblock In {\em In FOCS}, pages 671--680. IEEE, 2009.

\bibitem{rkiyeruai2012}
R.~Iyer and J.~Bilmes.
\newblock Algorithms for approximate minimization of the difference between
  submodular functions, with applications.
\newblock {\em In UAI}, 2012.

\bibitem{rkiyersubmodBregman2012}
R.~Iyer and J.~Bilmes.
\newblock The submodular {B}regman and {L}ov\'asz-{B}regman divergences with
  applications.
\newblock In {\em NIPS}, 2012.

\bibitem{iyermirrordescent}
R.~Iyer, S.~Jegelka, and J.~Bilmes.
\newblock {Mirror descent like algorithms for submodular optimization}.
\newblock {\em NIPS Workshop on Discrete Optimization in Machine Learning
  (DISCML)}, 2012.

\bibitem{jegelkathesis}
S.~Jegelka.
\newblock {\em Combinatorial Problems with submodular coupling in machine
  learning and computer vision}.
\newblock PhD thesis, ETH Zurich, 2012.

\bibitem{jegelka2011-inference-gen-graph-cuts}
S.~Jegelka and J.~A. Bilmes.
\newblock Approximation bounds for inference using cooperative cuts.
\newblock In {\em ICML}, 2011.

\bibitem{jegelka2011-nonsubmod-vision}
S.~Jegelka and J.~A. Bilmes.
\newblock Submodularity beyond submodular energies: coupling edges in graph
  cuts.
\newblock In {\em CVPR}, 2011.

\bibitem{jegelkanips}
S.~Jegelka, H.~Lin, and J.~Bilmes.
\newblock On fast approximate submodular minimization.
\newblock {\em In NIPS}, 2011.

\bibitem{krause2010sfo}
A.~Krause.
\newblock {SFO}: A toolbox for submodular function optimization.
\newblock {\em JMLR}, 11:1141--1144, 2010.

\bibitem{krause2008near}
A.~Krause, A.~Singh, and C.~Guestrin.
\newblock Near-optimal sensor placements in {G}aussian processes: Theory,
  efficient algorithms and empirical studies.
\newblock {\em JMLR}, 9:235--284, 2008.

\bibitem{kulesza2012determinantal}
A.~Kulesza and B.~Taskar.
\newblock Determinantal point processes for machine learning.
\newblock {\em arXiv preprint arXiv:1207.6083}, 2012.

\bibitem{lee2009non}
J.~Lee, V.S. Mirrokni, V.~Nagarajan, and M.~Sviridenko.
\newblock Non-monotone submodular maximization under matroid and knapsack
  constraints.
\newblock In {\em STOC}, pages 323--332. ACM, 2009.

\bibitem{matroidimproved}
Jon Lee, Maxim Sviridenko, and Jan Vondr{\'a}k.
\newblock Submodular maximization over multiple matroids via generalized
  exchange properties.
\newblock {\em In APPROX}, 2009.

\bibitem{lin2009select}
H.~Lin and J.~Bilmes.
\newblock How to select a good training-data subset for transcription:
  Submodular active selection for sequences.
\newblock In {\em Interspeech}, 2009.

\bibitem{linbudget}
H.~Lin and J.~Bilmes.
\newblock Multi-document summarization via budgeted maximization of submodular
  functions.
\newblock {\em In NAACL}, 2010.

\bibitem{linacl}
H.~Lin and J.~Bilmes.
\newblock A class of submodular functions for document summarization.
\newblock {\em In ACL}, 2011.

\bibitem{lin11}
H.~Lin and J.~Bilmes.
\newblock Optimal selection of limited vocabulary speech corpora.
\newblock In {\em Interspeech}, 2011.

\bibitem{mccormick2005submodular}
S~Thomas McCormick.
\newblock Submodular function minimization.
\newblock {\em Discrete Optimization}, 12:321--391, 2005.

\bibitem{mclachlan1997algorithm}
G.J. McLachlan and T.~Krishnan.
\newblock {\em The EM algorithm and extensions}.
\newblock New York, 1997.

\bibitem{nagano2011}
K.~Nagano, Y.~Kawahara, and K.~Aihara.
\newblock Size-constrained submodular minimization through minimum norm base.
\newblock In {\em ICML}, 2011.

\bibitem{nagano10macc}
K.~Nagano, Y.~Kawahara, and S.~Iwata.
\newblock Minimum average cost clustering.
\newblock In {\em NIPS}, 2010.

\bibitem{narasimhanbilmes}
M.~Narasimhan and J.~Bilmes.
\newblock A submodular-supermodular procedure with applications to
  discriminative structure learning.
\newblock In {\em UAI}, 2005.

\bibitem{narasimhan2006q}
M.~Narasimhan, N.~Jojic, and J.~Bilmes.
\newblock Q-clustering.
\newblock {\em NIPS}, 18:979, 2006.

\bibitem{nemhauser1978}
G.L. Nemhauser, L.A. Wolsey, and M.L. Fisher.
\newblock An analysis of approximations for maximizing submodular set
  functions---i.
\newblock {\em Mathematical Programming}, 14(1):265--294, 1978.

\bibitem{orlin2009faster}
J.B. Orlin.
\newblock A faster strongly polynomial time algorithm for submodular function
  minimization.
\newblock {\em Mathematical Programming}, 118(2):237--251, 2009.

\bibitem{rousu2006efficient}
J.~Rousu and J.~Shawe-Taylor.
\newblock Efficient computation of gapped substring kernels on large alphabets.
\newblock {\em Journal of Machine Learning Research}, 6(2):1323, 2006.

\bibitem{stobbe10efficient}
P.~Stobbe and A.~Krause.
\newblock Efficient minimization of decomposable submodular functions.
\newblock In {\em NIPS}, 2010.

\bibitem{sviridenko2004note}
M.~Sviridenko.
\newblock A note on maximizing a submodular set function subject to a knapsack
  constraint.
\newblock {\em Operations Research Letters}, 32(1):41--43, 2004.

\bibitem{svitkina2008submodular}
Z.~Svitkina and L.~Fleischer.
\newblock Submodular approximation: Sampling-based algorithms and lower bounds.
\newblock In {\em FOCS}, pages 697--706, 2008.

\bibitem{wan02networks}
P.-J. Wan, G.~Calinescu, X.-Y. Li, and O.~Frieder.
\newblock Minimum-energy broadcasting in static ad hoc wireless networks.
\newblock {\em Wireless Networks}, 8:607--617, 2002.

\bibitem{yuille2002concave}
A.L. Yuille and A.~Rangarajan.
\newblock The concave-convex procedure ({CCCP}).
\newblock {\em In NIPS}, 2002.

\end{thebibliography}
\end{document}